\documentclass[11pt]{amsart}

\usepackage{amsmath}
\usepackage{amssymb,amsfonts}
\usepackage{enumerate}
\usepackage[utf8]{inputenc}
\usepackage[margin=2.5cm]{geometry}

\usepackage[pdfborder={0 0 0}]{hyperref}
\usepackage{todonotes}
\usepackage[capitalize]{cleveref}
\usepackage{verbatim}

\newtheorem{theorem}{Theorem}

\newtheorem{lemma}[theorem]{Lemma}
\newtheorem{proposition}[theorem]{Proposition}
\newtheorem{corollary}[theorem]{Corollary}

\theoremstyle{remark}
\newtheorem{remark}[theorem]{Remark}
\newtheorem{example}[theorem]{Example}
\theoremstyle{definition}
\newtheorem{algorithm}[theorem]{Algorithm}

\newcommand{\nxt}[1]{\bar{#1}}
\newcommand{\nxte}[2]{#1_{#2}}

\newcommand{\funfrdist}[1]{\delta^{#1}_{\mathrm{FR}}}
\newcommand{\frdist}[3]{\funfrdist{#1}(#2;#3)}
\newcommand{\simfrdist}[1]{\funfrdist{2}(#1)}
\newcommand{\funsimfrdist}{\funfrdist{2}}

\newcommand{\frnumber}[2]{\mathrm{E}\left(#1,#2\right)}
\newcommand{\simfrnumber}[1]{\frnumber{#1}{2}}

\newcommand{\setdiv}[2]{\mathrm{D}_{#1}(#2)}
\newcommand{\simsetdiv}[1]{\mathrm{D}(#1)}
\newcommand{\cardsetdiv}[2]{\#\setdiv{#1}{#2}}
\newcommand{\cardsimsetdiv}[1]{\#\simsetdiv{#1}}

\newcommand{\apery}[2]{\mathrm{Ap}(#1,#2)}
\newcommand{\cardapery}[2]{\#\apery{#1}{#2}}

\newcommand{\finset}[1]{\{#1\}}

\newcommand{\setpossibles}[2]{\mathcal{P}^{#1}(#2)}
\newcommand{\finpossibles}[2]{\mathcal{X}^{#1}(#2)}

\title[On the second Feng-Rao distance of AG codes related to Arf semigroups]{On the second Feng-Rao distance of Algebraic Geometry codes related to Arf semigroups}

\author{Jos\'{e} I. Farr\'{a}n}
\address{Departamento de Matem\'{a}tica Aplicada, Universidad de Valladolid, Escuela de Ingenier\'{\i}a Inform\'{a}tica de Segovia, Espa\~{n}a}
\email{jifarran@eii.uva.es}
\thanks{The first author is supported by the project MTM2015-65764-C3-1-P (MINECO/FEDER)}

\author{Pedro A. Garc\'{\i}a-S\'anchez}
\address{IEMath-GR and Departamento de \'Algebra, Universidad de Granada, E-18071 Granada, Espa\~na}
\email{pedro@ugr.es} 
\thanks{The second author is supported by the projects MTM2014-55367-P, FQM-343,  FQM-5849, and FEDER funds}

\author{Benjam\'{\i}n A. Heredia}
\address{Departamento de Matem\'atica e Centro de Matem\'atica e Aplica\c{c}oes (CMA), FCT, Universidade Nova de Lisboa}
\email{b.heredia@fct.unl.pt}
\thanks{The third author is supported by the Funda\c{c}\~ao para a Ci\^encia e a Tecnologia (Portuguese Foundation for Science and Technology) through the project UID/MAT/00297/2013 (Centro de Matem\'atica e Aplica\c{c}\~oes).}

\begin{document}


\maketitle

\begin{abstract}
	
	We describe the second (generalized) Feng-Rao distance for elements in an Arf numerical semigroup that are greater than or equal to the conductor of the semigroup. This provides a lower bound for the second Hamming weight for one point AG codes. In particular, we can obtain the second Feng-Rao distance for the codes defined by asymptotically good towers of function fields whose Weierstrass semigroups are inductive. In addition, we compute the second Feng-Rao number, and provide some examples and comparisons with previous results on this topic. These calculations rely on Ap\'{e}ry sets, and thus several results concerning Ap\'ery sets of Arf semigroups are presented.
\end{abstract}

\section{Introduction}

Algebraic Geometry codes (AG codes for short) are a family of error-correcting codes whose parameters improve the Gilbert-Varshamov codes (see~\cite{HvLP}). These codes can be efficiently decoded by means of the so-called Feng-Rao majority voting decoding algorithm introduced in~\cite{f-r}. This algorithm decodes up to half the so-called Feng-Rao distance, which is a lower bound for the minimum distance, better than the Goppa distance. The Feng-Rao distance depends only on a Weierstrass semigroup of the underlying algebraic curve, and this allows us to study this parameter in general numerical semigroups. 

The construction of asymptotically good sequences of AG codes became more explicit from the introduction of asymptotically good towers of function fields by Garc\'{\i}a and Stichtenoth \cite{GS}. Numerical semigroups associated to this construction are inductive, and in particular are Arf semigroups (see \cite{Arf-IEEE} and \cite{inductivos}). 

The minimum distance was extended with the introduction of the generalized Hamming weight, independently in \cite{HKM} and \cite{Wei} for coding and cryptographic purposes, respectively. The generalization of the Feng-Rao bound was first considered in~\cite{H-P}, where it was proven that the so-called generalized (Feng-Rao) order bounds are also lower bounds for the generalized Hamming weights. 

For large elements of the underlying numerical semigroup $\Gamma$, the $r^{th}$ (generalized) Feng-Rao distance equals the Goppa bound up to a number, which is called the $r^{th}$ Feng-Rao number, denoted $\mathrm E(\Gamma,r)$. This integer depends solely on $r$ and the semigroup $\Gamma$. The second Feng-Rao number for inductive semigroups was computed in~\cite{inductivos}, and as a consequence, that of the tower of function fields by Garc\'{\i}a and Stichtenoth. Nevertheless, for small elements of the semigroup, namely those between the conductor $c$ and $2c-1$, the actual generalized Feng-Rao distance can be larger than the corresponding Goppa-like bound. This is due to the fact that Arf semigroups are (almost always) not symmetric (see~\cite{F-M}). 

The problem of computing Feng-Rao numbers and generalized Feng-Rao distances is hard. For numerical semigroups generated by intervals or by two (coprime) positive integers we have formulas for the $r^{th}$ Feng-Rao numbers (see \cite{intervalos} and \cite{D}, respectively). In this paper we compute the second Feng-Rao distance in the whole interval $[c,2c-1]$ for Arf semigroups, generalizing the results of both \cite{Arf-IEEE} and \cite{inductivos}. The computation is done by means of a very explicit algorithm, and the calculations improve sensibly those obtained in \cite{inductivos}. 

Every Arf numerical semigroup can be constructed by a series of translations of the form $\nxte{\Gamma}{m}=\{0\}\cup(m+\Gamma)$, starting with $\mathbb N$ (the set of nonnegative integers). We study how Ap\'ery sets and sets of divisors of $\nxte{\Gamma}{m}$ can be constructed from the corresponding sets in $\Gamma$ (more details in Section~\ref{sec:preliminaries}). 
With this machinery we will be able to infer properties of the second Feng-Rao distance and Feng-Rao numbers of an Arf numerical semigroup. 

The approach we follow in this paper is different from the one given in \cite{inductivos}, where homothecy was considered instead. The class studied there was the one closed under these homothecies: the class of inductive numerical semigroups. Every inductive numerical semigroup has the Arf property; thus in some sense the results presented in this manuscript extend those given in \cite{inductivos}. We also study here Feng-Rao distances, which were not considered in \cite{inductivos}. In particular, the second Feng-Rao distance of the inductive semigroups underlying in the construction \cite{GS} of Garc\'{\i}a and Stichtenoth can be explicitly computed with our algorithm. 

The paper is organized as follows: Section~\ref{sec:preliminaries} introduces the preliminary results needed in the rest of the manuscript concerning numerical semigroups, Arf semigroups, divisors, Ap\'{e}ry sets, Feng Rao numbers and generalized Feng-Rao distances. Section~\ref{sec:apery} computes the second Feng-Rao number for Arf semigroups. Section~\ref{sec:frdist} is devoted to the calculation of the second Feng-Rao distance in the interval $[c,2c-1]$, which is divided in two parts. We first study the interval $[c,c+e-1]$, with $e$ the multiplicity of the semigroup, and then the remaining elements in $[c,2c-1]$ by using an iterative procedure. Section~\ref{sec:hyperE} studies the particular case of hyperelliptic semigroups. 
Section~\ref{sec:computation} discusses some computational issues concerning generalized Feng-Rao distances for arbitrary semigroups, and finally Section \ref{sec:conclusiones} shows some examples and applications to AG codes coming from towers of function fields.

\section{Preliminaries}\label{sec:preliminaries}

In this section we introduce the notations and definitions needed in the rest of the paper. A \emph{numerical semigroup} is a set of nonnegative integers that is closed under addition, contains the zero, and has finite complement in $\mathbb N$. 

Let \(\Gamma=\{\rho_1,\rho_2,\ldots\}\) be a numerical semigroup, with $0=\rho_1<\rho_2<\ldots$. We say that $\Gamma$ is an \emph{Arf numerical semigroup}  if for every \(i\geq j\geq k\), \(\rho_i+\rho_j-\rho_k\) is in \(\Gamma\). Thus Arf numerical semigroups are a particular example of numerical semigroups that can be defined by a pattern \cite{patterns}. Arf numerical semigroups can be also characterized as numerical semigroups attaining a redundancy bound associated to evaluation codes, \cite{maria1}. 

We will denote by \(e=\mathrm e(\Gamma)=\rho_2\) the \emph{multiplicity }of \(\Gamma\). The \emph{conductor} of \(\Gamma\), \(c=\mathrm c(\Gamma)=\rho_r\), is the least integer $c$ such that $c+\mathbb N\subseteq \Gamma$. The \emph{genus} of \(\Gamma\) is the cardinality of $\mathbb N\setminus \Gamma$, and we will refer to it as \(g=\mathrm g(\Gamma)\). More details on numerical semigroups can be found in \cite{ns}, and some applications in \cite{a-g-b}. A nice review of the interaction between numerical semigroups and AG codes can be found in \cite{codes-ns}.

For a numerical semigroup $\Gamma$, and $x,y\in \mathbb Z$, we say that $x$ \emph{divides} $y$ if $y-x\in \Gamma$ (in the literature this is denoted sometimes by $x\le_\Gamma y$). This relation is an ordering in $\mathbb Z$. If $x\in\Gamma$ and $x\neq 0$, then $x$ has at least two trivial divisors ($0$ and $x$ itself), and $x$ is called \emph{irreducible} if it has exactly these two trivial divisors. 

A subset $A$ of $\Gamma$ is a \emph{generating system} of $\Gamma$ if $\Gamma=\langle a_1+\cdots +a_k\mid k\in \mathbb N, a_1,\ldots,a_k\in A\rangle$. Every numerical semigroup $\Gamma$ has a \emph{minimal generating system}, that is, a generating system such that none of its subsets generates the semigroup. This minimal generating system is precisely $\Gamma^*\setminus (\Gamma^*+\Gamma^*)$, with $\Gamma^*=\Gamma\setminus \{0\}$ (we will use the asterisk notation to remove the zero element from a set of integers). The elements of the minimal generating system are precisely the \emph{irreducible elements} of the semigroup. The cardinality of the minimal generating system of a numerical semigroup $\Gamma$ is known as the \emph{embedding dimension} of $\Gamma$. As two irreducible elements in a numerical semigroup cannot be congruent modulo the multiplicity of the semigroup, it follows that the embedding dimension is less than or equal to the multiplicity of the numerical semigroup. Numerical semigroups with \emph{maximal embedding dimension} are thus numerical semigroups with embedding dimension equal to the multiplicity. It is well known that numerical semigroups with the Arf property have maximal embedding dimension (see for instance \cite[Chapter 2]{ns}).

For any integer \(x\in \mathbb{Z}\), denote the \emph{Apéry set} of \(\Gamma\) with respect to \(x\) by
\[
\apery{\Gamma}{x}=\finset{m\in \Gamma \mid m - x \notin \Gamma}.
\]
The Ap\'ery set of $x\in \mathbb Z$ is formed precisely by the elements in $\Gamma$ that are not 
\lq\lq divisible\rq\rq \ (with respect to the semigroup) by $x$. If $x\in \Gamma$, then \(\apery{\Gamma}{x}\) has precisely $x$ elements (indeed the converse is also true). Ap\'ery in \cite{apery} originally defined these sets only for elements in the semigroup. We will see later that extending this definition to every integer is quite convenient.

If $\Gamma$ has maximal embedding dimension, then the set
\[
\{e\}\cup(\apery{\Gamma}{e}\setminus\{0\})
\]
is the minimal generating system of $\Gamma$ (see for instance \cite[Chapter 2]{ns}).

Let $\Gamma$ be a numerical semigroup. The \emph{set of divisors} in $\Gamma$ of $x\in \mathbb Z$ is given by
\[
\simsetdiv{x}=\setdiv{\Gamma}{x}=\finset{m\in \Gamma \mid x-m \in \Gamma}.
\]
It is easy to see that this set is not empty if and only if \(x\in \Gamma\). If \(m\in \Gamma\), then \(0,m\in \simsetdiv{m}\) and we have \(\cardsimsetdiv{m}=2\) if and only if \(m\) is irreducible. The following will be useful later.
\begin{lemma}\label{divisorsContained}
	Let $\Gamma$ be a numerical semigroup. Given \(m,m'\in \Gamma\) with \(m\leq m'\) we have that \(\simsetdiv{m}\subseteq \simsetdiv{m'}\) if and only if \(m\in \simsetdiv{m'}\).
\end{lemma}
\begin{proof}
	Since \(m\in \simsetdiv{m}\) the necessity is obvious. 
	
	For the converse, suppose that \(m\in \simsetdiv{m'}\) and take \(x\in \simsetdiv{m}\). Then \(x\in \Gamma\) and \(m'-x=(m'-m) + (m-x)\in \Gamma\), whence \(m\in \simsetdiv{m'}\).
\end{proof}

For \(m_1,\ldots,m_k\in \Gamma\), we write
\[
\simsetdiv{m_1,\ldots,m_k}=\setdiv{\Gamma}{m_1,\ldots,m_k}=
\simsetdiv{m_1}\cup\cdots \cup \simsetdiv{m_k}.
\]

The \emph{\(r^\text{th}\) Feng-Rao distance} of $m\in \Gamma$ is given by 
\[
\frdist{r}{\Gamma}{m}=\min\finset{\cardsetdiv{\Gamma}{m_1,\ldots,m_r}\mid m\leq m_1<\cdots<m_r, m_i\in \Gamma}.
\]
One of the goals of this paper is to compute \(\frdist{2}{\Gamma}{m}\) for $\Gamma$ an Arf numerical semigroup and $m\ge c$ (the conductor of $\Gamma$).

For \(m\geq 2c-1\) we have
\[
\frdist{r}{\Gamma}{m}=m+1 - 2g + \frnumber{\Gamma}{r},
\]
for some $\frnumber{\Gamma}{r}$ depending only on $\Gamma$ and $r$ (see for instance \cite{F-M}). 

Moreover, we also have 
\[
\frdist{r}{\Gamma}{m} \geq m+1 - 2g + \frnumber{\Gamma}{r}
\]
for $m\geq c$. Note that the case $m<c$ does not make sense for AG codes, since we should have $m>2g-2$ if we want to have an injective coding map (see~\cite{HvLP}). It may happen that the above inequality becomes an equality for all integers greater than a certain bound less than $2c-1$. For instance, for $r=1$, such a bound has been calculated in \cite{maria-acute} for acute semigroups, and Arf numerical semigroups are acute.

For the particular case $r=2$, we have
\begin{equation}\label{frnumberwithApery}
	\simfrnumber{\Gamma}=\min\finset{\cardapery{\Gamma}{x} \mid 1\leq x \leq e(\Gamma)}
\end{equation}
(see \cite{F-M}, or \cite{inductivos} with the notation used here).

Given a numerical semigroup \(\Gamma\) and an element \(m\in \Gamma\), the set
\[
\nxte{\Gamma}{m}=\finset{0}\cup(m+\Gamma)
\]
is again a numerical semigroup (indeed it has maximal embedding dimension, \cite[Chapter 2]{ns}). Moreover, $\Gamma$ has the Arf property if and only if $\nxte{\Gamma}{m}$ has the Arf property \cite{pi-semigrupos}. 


The following result follows immediately from the definition of $\nxte{\Gamma}{m}$, and will be used later without referencing to it.

\begin{proposition}
	Let $\Gamma$ be a numerical semigroup and let $m\in \Gamma$. Then 
	\begin{itemize}
		\item $\mathrm g(\Gamma_m)=\mathrm g(\Gamma)+m-1$,
		\item $\mathrm c(\Gamma_m)=\mathrm c(\Gamma)+m$,
		\item $\mathrm e(\Gamma_m)=m$.
	\end{itemize}
\end{proposition}

We will prove later that for Arf numerical semigroups, $\mathrm E(\nxte{\Gamma}{m},2) = \min \{m,\mathrm E(\Gamma,2)+1\}$.

In order to simplify notation, $\nxte{(\nxte{\Gamma}{m})}{m'}$ will be denoted by $\nxte{\Gamma}{m,m'}$.

For $i$ a positive integer, set \[d_i := \rho_{i+1}-\rho_i,\] the distance between two consecutive elements in \(\Gamma\).

\begin{lemma}\label{lema1}
	Let $\Gamma$ be an Arf numerical semigroup. If \(i\leq j\), then we have \(\rho_j+d_i\in \Gamma\) and $d_i\ge d_j$.
\end{lemma}
\begin{proof}
	The first assertion is a consequence of the Arf property for the triple \(j\geq i+1 > i\), if \(j\geq i+1\), and \(i+1>j\geq i\) if \(j=i\).
	
	If $i\le j$, as $\rho_j+d_i\in \Gamma$ and $\rho_j+d_i> \rho_j$, we obtain $\rho_j+d_i\ge \rho_{j+1}$. Thus $d_i\ge \rho_{j+1}-\rho_j=d_j$.
\end{proof}

In particular, \(d_r\le \dots \le d_2 \leq d_1 =e=\rho_2\). Also $d_{r+k}=1$ for all $k\in\mathbb N$ ($d_r=\rho_{r+1}-\rho_r=c+1-c=1$). 

The sequence $(d_1,\ldots,d_r)$ is known in the literature as the \emph{multiplicity sequence} of $\Gamma$ (see for instance \cite{arf-algebroid} or \cite{param-Arf}), and if $\Gamma$ is an Arf numerical semigroup, it determines the semigroup since
\[
\Gamma = \finset{0,d_1,d_1+d_2,\ldots, d_1+\cdots +d_{r-1},\to}
\]
(the arrow here means that all integers greater than the integer preceding it are in the set).

\begin{remark}\label{Arf-iterative}
	Let $\Gamma$ be an Arf numerical semigroup, and let $(d_1,\ldots,d_r)$ be its multiplicity sequence. Then 
	\begin{equation}\label{nxtWithdi}
		\Gamma=\nxte{\mathbb N}{d_{r-1},\ldots,d_1}=\nxte{(\nxte{\mathbb N}{d_{r-1},\ldots, d_2})}{d_1}.
	\end{equation}
	In other words, every Arf semigroup can be obtained from $\mathbb N$ after a finite set of translations (and adding 0 to become a monoid). 
\end{remark}

\begin{remark}\label{Inductive}
	Let \(\Gamma\) be a numerical semigroup with conductor \(c\), and let \(m\in \Gamma\). For any \(a,b\in \mathbb{Z}\) such that \(a\geq 2\) and \(b\geq c\) we have
	\[
	a\nxte{\Gamma}{m}\cup(a(b+m)+\mathbb{N})=\nxte{(a\Gamma \cup(ab+\mathbb{N}))}{am}.
	\]
	
	As a consequence, if \((d_1,\ldots,d_{r})\) is the multiplicity sequence of an Arf numerical semigroup \(\Gamma\) with conductor \(c\), and \(a,b\in \mathbb{Z}\) such that \(a\geq 2\) and \(b\geq c\), then the set \(a\Gamma \cup (ab + \mathbb{N})\) is an Arf numerical semigroup with multiplicity sequence \((ad_1,\ldots,ad_{r-1},a,\ldots,a,1)\) where \(a\) appears \(b-c\) times in the sequence. That is,
	\[
	a\Gamma \cup (ab+\mathbb{N})= \nxte{\mathbb{N}}{a,\stackrel{b-c}{\ldots},a,ad_{r-1},\ldots,ad_1}.
	\]
	Starting with \(\Gamma=\mathbb{N}\), and repeating this process we get the class of inductive numerical semigroups, see \cite{inductivos}.
\end{remark}

\section{Apéry sets and the second Feng-Rao number}\label{sec:apery}

In light of \cref{frnumberwithApery}, a good understanding of Ap\'ery sets will help us in the computation of the second Feng-Rao number. 

Let $\Gamma$ be a numerical semigroup and let $0\neq \nxt{e}\in \Gamma$. Recall that \(\nxte{\Gamma}{\nxt{e}} = \finset{0}\cup (\nxt{e}+\Gamma)\). It is clear that \(\nxte{\Gamma}{\nxt{e}}\subset \Gamma\).
Let us denote \(\apery{\Gamma}{\nxt{e}}^*=\apery{\Gamma}{\nxt{e}}\setminus\finset{0}\).

\begin{proposition}
	Let $\Gamma$ be a numerical semigroup and let $\nxt{e}\in \Gamma^*$. Then  
	\[\Gamma\setminus \nxte{\Gamma}{\nxt{e}}=\apery{\Gamma}{\nxt{e}}^*.\]
\end{proposition}
\begin{proof}
	Observe that \(m\in \Gamma\setminus\nxte{\Gamma}{\nxt{e}}\) if and only if $m\in \Gamma$ and \(m\notin \nxt{e}+\Gamma\), which happens if and only if $m\in \Gamma$ and \(m-\nxt{e}\notin \Gamma\).
\end{proof}

Every nonzero element of $\nxte{\Gamma}{\nxt{e}}$ is of the form $\nxt{e}+m$ for some $m\in \Gamma$. The following result describes the Ap\'ery set of $\nxt{e}+m$ in $\nxte{\Gamma}{\nxt{e}}$ in terms of the Ap\'ery sets of $\nxt{e}$ and $m$ in $\Gamma$. 

\begin{proposition}
	Let $\Gamma$ be a numerical semigroup and let $\nxt{e}$ be a nonzero element of $\Gamma$. For any \(m\in \Gamma\), we have
	\[
	\apery{\nxte{\Gamma}{\nxt{e}}}{\nxt{e}+m}=
	(\nxt{e}+\apery{\Gamma}{m})\cup (\nxt{e}+m+\apery{\Gamma}{\nxt{e}}^*)\cup \finset{0}.
	\]
\end{proposition}
\begin{proof}
	Let \(m'\in \apery{\nxte{\Gamma}{\nxt{e}}}{\nxt{e}+m}^*\). Then \(m'=\nxt{e}+s'\) with \(s'\in\Gamma\). Now \(\nxt{e}+s'-(\nxt{e}+m)= s'-m \notin \nxte{\Gamma}{\nxt{e}}\), so that we have two possibilities:
	\begin{itemize}
		\item if \(s'-m\in \Gamma\), then \(s'-m\in \Gamma\setminus\nxte{\Gamma}{\nxt{e}}=\apery{\Gamma}{\nxt{e}}^*\), whence \(s'\in m+\apery{\Gamma}{\nxt{e}}^*\);
		\item if \(s'-m\notin \Gamma\), then \(s'\in \apery{\Gamma}{m}\).
	\end{itemize}
	
	For the other inclusion, if \(s' \in \apery{\Gamma}{m}\), then \(\nxt{e}+s'\in \nxte{\Gamma}{\nxt{e}}\) and \(\nxt{e}+s'-(\nxt{e}+m)=s'-m\notin \Gamma\). In particular, \(\nxt{e}+s'-(\nxt{e}+m)\notin \nxte{\Gamma}{\nxt{e}}\). So \(\nxt{e}+\apery{\Gamma}{m}\subset \apery{\nxte{\Gamma}{\nxt{e}}}{\nxt{e}+m}\).
	
	Finally, let \(s'\in \apery{\Gamma}{\nxt{e}}^*=\Gamma\setminus \nxte{\Gamma}{\nxt{e}}\). Then \(\nxt{e}+m+s'\in \nxte{\Gamma}{\nxt{e}}\), but \(\nxt{e}+m+s'-(\nxt{e}+m)=s'\notin \nxte{\Gamma}{\nxt{e}}\). Hence \(\nxt{e}+m+\apery{\Gamma}{\nxt{e}}^*\subset \apery{\nxte{\Gamma}{\nxt{e}}}{\nxt{e}+m}\).
\end{proof}

Recall that the second Feng-Rao number can be computed as the minimum of the cardinalities of the Ap\'ery sets of positive integers less than the multiplicity. We now give a series of lemmas that will allow to see how these cardinalities behave in  an Arf numerical semigroup.

\begin{lemma}\label{lema3}
	Let $\Gamma$ be an Arf numerical semigroup. If \(d_i< d_{i-1}\), then
	\[
	\apery{\Gamma}{-d_i} = \finset{\rho_1,\ldots,\rho_{i-1}}.
	\]
\end{lemma}
\begin{proof}
	For any \(j\leq i-1\) we have \(\rho_j<\rho_j + d_i<\rho_j + d_j=\rho_{j+1}\), so \(\rho_j+d_i \notin \Gamma\), which means that \(\rho_j\) is in the Apéry set. \cref{lema1} ensures that if \(j\geq i\), \(\rho_j\notin \apery{\Gamma}{-d_i}\).
\end{proof}

By using the relationship between the Ap\'ery set of an integer and its opposite, we can get the cardinality of the Ap\'ery sets of the distances between elements in an Arf numerical semigroup.

\begin{corollary}\label{cor1}
	Let $\Gamma$ be an Arf numerical semigroup. If \(d_i<d_{i-1}\), then
	\[
	\cardapery{\Gamma}{d_i} = i-1+d_i.
	\]
\end{corollary}
\begin{proof}
	We only need to use the fact that \(\cardapery{\Gamma}{x}=\cardapery{\Gamma}{-x} + x\) (see \cite[Lemma 1]{FGHL}) together with \cref{lema3}.
\end{proof}

Next we see that the cardinality of the Ap\'ery sets of integers between the distances are controlled by the distances.

\begin{lemma}
	Let $\Gamma$ be an Arf numerical semigroup. For any \(d_j\leq x < d_{j-1}\) we have
	\[
	\apery{\Gamma}{-d_j} \subseteq \apery{\Gamma}{-x}.
	\]
	In particular, \(\cardapery{\Gamma}{x}\ge \cardapery{\Gamma}{d_j}\).
\end{lemma}
\begin{proof}
	For any \(k\leq j-1\), we have \(\rho_k<\rho_k +x <\rho_k + d_k=\rho_{k+1}\), and so \(\rho_k-(-x)=\rho_k+x \notin \Gamma\). By \cref{lema3}, this implies that \(\apery{\Gamma}{-d_j} \subseteq \apery{\Gamma}{-x}\), and thus by \cite[Lemma 1]{FGHL},
	\[
	\cardapery{\Gamma}{x}=\cardapery{\Gamma}{-x}+x\geq \cardapery{\Gamma}{-d_j}+d_j=\cardapery{\Gamma}{d_j}.\qedhere
	\]
\end{proof}

\begin{remark}\label{wherefrnumber}
	This means that in order to compute the second Feng-Rao number using the formula in \cref{frnumberwithApery} we only need to check \(\apery{\Gamma}{d}\), with \(d\) a distance between consecutive elements in \(\Gamma\).
\end{remark}

\begin{remark}\label{consecutivos}
	Notice that if $\Gamma$ is an Arf numerical semigroup and $m,m+1\in\Gamma$, then $2m+2-m=m+2\in \Gamma$. Arguing in this way, $m\ge c$. This proves the following result.
\end{remark}

\begin{lemma}\label{aperyin1}
	\(\apery{\Gamma}{1}=\Gamma\cap [0,c]=\finset{0=\rho_1,\rho_2,\ldots,\rho_r=c}\).
\end{lemma}

With this, we can calculate the second Feng-Rao number of an Arf numerical semigroup in terms of its multiplicity sequence.

\begin{proposition}\label{frnumber1}
	Let $\Gamma$ be an Arf numerical semigroup with multiplicity sequence $(d_1,\ldots, d_r)$, that is, \(\Gamma=\nxte{\mathbb N}{d_{r-1},\ldots,d_1}\). Then
	\[
	\simfrnumber{\nxte{\mathbb N}{d_{r-1},\ldots,d_1}}=\min\finset{d_1,d_2+1,\ldots, d_{r-1}+r-2,r}.
	\]
\end{proposition}
\begin{proof}
	Using \cref{cor1}, we have
	\[
	\cardapery{\Gamma}{d_i}=
	\begin{cases}
	\cardapery{\Gamma}{d_{i-1}} & \text{ if } d_i=d_{i-1},\\
	d_i+i-1 & \text{ if } d_i<d_{i-1}.
	\end{cases}
	\]
	So if \(d_{j+1}<d_{j}=d_i<d_{i-1}\), with \(i\leq j\), we obtain
	\[
	d_j+j-1\geq d_i+i-1 = \cardapery{\Gamma}{d_{j}}=\cardapery{\Gamma}{d_i}.
	\]
	Since by \cref{aperyin1}
	\[
	\cardapery{\Gamma}{d_r}=\cardapery{\Gamma}{1}=r,
	\]
	we can apply \cref{wherefrnumber} to get the desired result.
\end{proof}
Also, this allows us to give a recursive formula. Recall that $\Gamma$ has the Arf property if and only if $\nxte{\Gamma}{\nxt{e}}$ has the Arf property.
\begin{theorem}\label{Arf-E2}
	Let $\Gamma$ be an Arf numerical semigroup and let $\nxt{e}\in \Gamma^*$. Then
	\[
	\simfrnumber{\nxte{\Gamma}{\nxt{e}}}=\min\finset{\nxt{e},\simfrnumber{\Gamma}+1}.
	\]
\end{theorem}
\begin{proof}
	This is a consequence of \cref{frnumber1} together with \cref{nxtWithdi}
\end{proof}

Notice that, by using Remark \ref{Arf-iterative}, the above theorem provides us with an algorithm to compute the second Feng-Rao number for any Arf numerical semigroup. In fact, if $\Gamma = \nxte{\mathbb{N}}{d_{r-1},\ldots,d_{1}}$ as in Section \ref{sec:preliminaries}, then we start with $\Gamma = \mathbb{N}$ where $\simfrnumber{\mathbb{N}}=1$, and iterate Theorem \ref{Arf-E2} with $m$ from $d_{r-1}$ to $d_{1}$. 

\begin{example}
	Let $\Gamma=\langle 5, 7, 9, 11, 13\rangle$, which has multiplicity sequence $( 5, 2, 2, 1)$. Observe that $\nxte{\mathbb N}{1}=\mathbb N$. By \cref{Arf-E2}, $\simfrnumber{\nxte{\mathbb N}{2}}=\min\{ 2, \simfrnumber{\mathbb N}+1\}=2$. By using again this formula, $\simfrnumber{\nxte{\mathbb N}{2,2}}=\min\{2, 2+1\}=2$. Finally, $\simfrnumber{\Gamma}=\simfrnumber{\nxte{\mathbb N}{2,2,5}}=\min\{5,2+1\}=3$.
	
	If we use \cref{frnumber1}, $\simfrnumber{\Gamma}=\min\{5,2+1,2+2,4\}=3$.
\end{example}

As a corollary of \cref{frnumber1} and \cref{Inductive} we obtain the following formula for the homothetic transformation (compare with \cite{inductivos}; recall that every inductive semigroup has the Arf property).

\begin{corollary}
	Let \(\Gamma\) be an Arf numerical semigroup with conductor \(c\). Let \(\nxt{e}\in \Gamma^*\), and \(a,b\in \mathbb{Z}\) such that \(a\geq 2\) and \(b\geq c\). Then
	\[
	\simfrnumber{a\nxte{\Gamma}{\nxt{e}}\cup (a(b+\nxt{e})+\mathbb{N})}=
	\min\finset{a \nxt{e},\simfrnumber{a\Gamma \cup (ab+\mathbb{N})}+1}.
	\]
	
	Furthermore, if \(\Gamma\) has multiplicity sequence \((d_1,\ldots, d_r)\), then
	\[
	\simfrnumber{a\Gamma\cup (ab+\mathbb{N})}=\min\finset{ad_1,ad_2+2,\ldots, ad_{r-1}+r-2, a+r-1, (b-c) + r}.
	\]
\end{corollary}

\section{The second Feng-Rao distance}\label{sec:frdist}

Our next goal is to compute the second Feng-Rao distance for elements in an Arf numerical semigroup greater than the conductor of the semigroup. To this end, we will apply again Remark \ref{Arf-iterative} iteratively. Note that for $m \geq 2c-1$ one has $\simfrdist{m}=m+1-2g+\simfrnumber{\Gamma}$ \cite{F-M} and we have computed the second Feng-Rao number for Arf semigroups. The starting case $\Gamma = \mathbb{N}$ is actually trivial, since $\simfrdist{m}=m+2$. 

Given an Arf numerical semigroup $\Gamma$ with conductor $c$, we first study the second Feng-Rao distance for the elements in $[0,c+e-1]\cap\Gamma$, and then from $c+e$ on. We need some lemmas relating the divisors of a numerical semigroup and its translate by one of its elements (as we had in the case of Ap\'ery sets). 

\begin{lemma}\label{divisorsLem}
	Let $\Gamma$ be a numerical semigroup and let $\nxt{e}\in \Gamma$. 
	For any \(m \in \Gamma\), we have
	\[
	\setdiv{\nxte{\Gamma}{\nxt{e}}}{2\nxt{e}+m}=(\nxt{e}+\setdiv{\Gamma}{m})\cup\finset{0,2\nxt{e}+m},
	\]
	and this union is disjoint.
\end{lemma}
\begin{proof}
	First, we show that \(\nxt{e}+\setdiv{\Gamma}{m}\subseteq \setdiv{\nxte{\Gamma}{\nxt{e}}}{2\nxt{e}+m}\). Take \(s\in \setdiv{\Gamma}{m}\). Then \(\nxt{e}+s\) is clearly in \(\nxte{\Gamma}{\nxt{e}}\), and \(2\nxt{e}+m-(\nxt{e}+s)=\nxt{e}+(m-s)\), which is in \(\nxte{\Gamma}{\nxt{e}}\) because \(m-s\in \Gamma\). This proves one inclusion.
	
	Let now \(s\) be an element in \(\setdiv{\nxte{\Gamma}{\nxt{e}}}{2\nxt{e}+m}\setminus\finset{0,2\nxt{e}+m}\). Since \(s\in \nxte{\Gamma}{\nxt{e}}\setminus\finset{0}\), we have \(s-\nxt{e}\in \Gamma\). Also, $m-(s-\nxt{e})=\nxt{e}+m-s$, and as $2\nxt{e}+m-s\in \nxte{\Gamma}{\nxt{e}}\setminus\{0\}$, we conclude that $\nxt{e}+m-s\in \Gamma$. Thus, $s=\nxt{e}+(s-\nxt{e})\in \nxt{e}+\setdiv{\Gamma}{m}$.
\end{proof}

From the above lemma we can determine completely the (classical) Feng-Rao distance for an Arf numerical semigroup (this was done already in \cite{Arf-IEEE}). 

\begin{theorem}[Campillo-Farrán-Munuera]\label{th:CFM}
	Let $\Gamma$ be an Arf numerical semigroup with $c=\rho_{r}$, and denote $m_{1}:=0$ and $m_{k}:=c+\rho_{k}-1$ for $k\in \{2,\ldots, r\}$. Then one has $\cardsetdiv{\Gamma}{m_{1}}=1$ and 
	$\cardsetdiv{\Gamma}{m_{k}}=2k-2$ for $k\in \{2,\ldots, r\}$, and hence 
	\begin{itemize}
		\item $\frdist{1}{\Gamma}{m_{1}}=1,$
		\item $\frdist{1}{\Gamma}{m}=2k-2$ if $m_{k-1}<m\leq m_{k}$, $k\in\{2,\dots,r\}$, 
		\item $\frdist{1}{\Gamma}{m}=m+1-2g$ if $m\ge m_{r}$. 
	\end{itemize}
\end{theorem}

Note that the main purpose of this paper is precisely to generalize the above result for the second Feng-Rao distance. Thus we need to know how divisors of two elements behave under translations.


\begin{lemma}\label{generalLema}
	Let $\Gamma$ be a numerical semigroup and let $\nxt{e}\in \Gamma$. For every configuration $m,m'\in\Gamma$, we have 
	\[
	\setdiv{\nxte{\Gamma}{\nxt{e}}}{2\nxt{e}+m,2\nxt{e}+m'} = \left( \nxt{e} + \setdiv{\Gamma}{m,m'} \right) \cup \finset{0,2\nxt{e}+m,2\nxt{e}+m'}.
	\]
	If in addition \(|m-m'|<  \nxt{e}\), then this union is disjoint. In particular,
	\[
	\cardsetdiv{\Gamma}{m,m'} +2 \le \cardsetdiv{\nxte{\Gamma}{\nxt{e}}}{2\nxt{e}+m,2\nxt{e}+m'} \le \cardsetdiv{\Gamma}{m,m'}+3.
	\]
\end{lemma}
\begin{proof}
	The first statement follows easily from \cref{divisorsLem}. For the second, suppose \(m\leq m'\), and so \(m'-m< \nxt{e}\). Then \(2\nxt{e}+m'\notin \left(\nxt{e}+\setdiv{\Gamma}{m}\right)\) because
	\(m-(\nxt{e}+m')=m-m'-\nxt{e}<0\). Similarly, \(2\nxt{e}+m\notin \left(\nxt{e} +\setdiv{\Gamma}{m'}\right)\), since \(m'-(\nxt{e}+m)=m'-m-\nxt{e}<0\).
	
	Notice that $0,2\nxt{e}+m'\not\in \nxt{e} + \setdiv{\Gamma}{m,m'}$. From this observation, the last inequality follows.
\end{proof}

The union might not be disjoint if \(|m-m'|\ge \nxt{e}\). For example, it is not disjoint for \(m'=m+\nxt{e}\).


The following result gives bounds for the cardinality of the set of divisors of two elements in a numerical semigroup in terms of the divisors of each element. This will be used later.

\begin{lemma}\label{boundsForNDiv}
	Let $\Gamma$ be a numerical semigroup. For any \(m,m'\in \Gamma\) with \(m< m'\) we have
	\[
	\max \finset{\cardsetdiv{\Gamma}{m}+1,\cardsetdiv{\Gamma}{m'}} \leq \cardsetdiv{\Gamma}{m,m'} \leq \cardsetdiv{\Gamma}{m}+\cardsetdiv{\Gamma}{m'}-1.
	\]
\end{lemma}
\begin{proof}
	The first inequality comes from the inclusions
	\[
	\setdiv{\Gamma}{m}\cup \finset{m'} \subseteq \setdiv{\Gamma}{m,m'},\quad 
	\setdiv{\Gamma}{m'} \subseteq \setdiv{\Gamma}{m,m'},
	\]
	where the first union is disjoint.
	
	The second inequality is clear since \(0\in \setdiv{\Gamma}{m}\cap \setdiv{\Gamma}{m'}\).
\end{proof}

Recall that we are looking for minimums of $\cardsetdiv{\Gamma}{m,m'}$ with $m_0,m,m'$ in a numerical semigroup $\Gamma$, with $m_0\le m<m'$. The following result is telling us that we can choose $m$ and $m'$ at distance at most the multiplicity of $\Gamma$. This simplifies the search.

\begin{lemma}\label{hastamultiplicidad}
	Let $\Gamma$ be a numerical semigroup with conductor $c$ and multiplicity \(e\). For \(m\geq c\), we have
	\[
	\min \finset{\cardsetdiv{\Gamma}{m,m'}\mid m< m'} = \min \finset{\cardsetdiv{\Gamma}{m,m'}\mid m<m' \leq m+e}.
	\]
\end{lemma}
\begin{proof}
	Suppose \(m'> m\) is the smallest element in which the minimum in the left hand side of the above equation is attained. If \(m'\geq m+e\), then \(m'-e\in \setdiv{\Gamma}{m'}\), and by \cref{divisorsContained} we have \(\setdiv{\Gamma}{m'-e}\subsetneq \setdiv{\Gamma}{m'}\). Thus,
	\[
	\cardsetdiv{\Gamma}{m,m'-e}< \cardsetdiv{\Gamma}{m,m'},
	\]
	which is in contradiction with the minimality of \(m'\).
\end{proof}

With this series of lemmas we are now ready to study the second Feng-Rao distance on an Arf numerical semigroup $\Gamma$. We will study first the interval $[0,c+e-1]\cap\Gamma$ and then the elements larger than $c+e$, with $c$ the conductor of $\Gamma$ and $e$ its multiplicity. This distinction will become clear later, since the results obtained in each case are different in nature.

\subsection{Second Feng-Rao distance up to \(c+e-1\)}

First, we study the case \(e=2\). In this setting, \(c+e-1=c+1\), $c$ is even, and $\Gamma=\langle e,c+1\rangle$. 

\begin{lemma}\label{caseMultiplicity2}
	Let $c$ be an even integer greater than one. For \(\Gamma=\langle 2,c+1\rangle\) and $m\in\Gamma$, we have
	\[
	\simfrdist{m}=\begin{cases} 3 & \text{ if } m=2, \\ 4  & \text { if } 2<m \leq c+1. \end{cases}
	\]
\end{lemma}
\begin{proof}
	It is clear that \(\simfrdist{2}=3\), since \(\setdiv{\Gamma}{2,c+1}=\{0,2,c+1\}\). For any other element of the semigroup with \(2<m\) and \(m\neq c+1\), \(\cardsetdiv{\Gamma}{m}\geq 3\), and by \cref{boundsForNDiv} \(\cardsetdiv{\Gamma}{m,m'}\geq 4\) for any \(m'>m\). As for \(m=c+1\), we have \(\finset{0,2,c+1,m'}\subset \setdiv{\Gamma}{c+1,m'}\) for any \(m'>c+1\), so also \(\cardsetdiv{\Gamma}{c+1,m'}\geq 4\). We deduce that \(\simfrdist{m}\geq 4\) for \(m>2\). Finally
	\[
	\setdiv{\Gamma}{c+1,c+3}=\finset{0,2,c+1,c+3},
	\]
	which gives \(\simfrdist{m}=4\) for \(2<m\leq c+1\).
\end{proof}

We now deal with the case of embedding dimension greater than two. The following property simplifies the task.

\begin{proposition}
	Let \(\Gamma\) be an Arf numerical semigroup with conductor \(c\) and multiplicity \(e\ge 3\). Then either \(c+e-2\) or \(c+e-3\) is irreducible.
\end{proposition}
\begin{proof}
	Recall that $\apery{\Gamma}{e}^*\cup\{e\}$ is the minimal generating system of $\Gamma$. 
	
	Since \(e\geq 3\), clearly both \(c+e-2\) and \(c+e-3\) are in \(\Gamma\). Now suppose that \(c+e-2\) is not irreducible. Then \(c+e-2\notin \apery{\Gamma}{e}^*\). That means that \(c-2\in \Gamma\). If \(c-3\in \Gamma\), then by Remark \ref{consecutivos},  we will have \(c-2\ge c\), which is a contradiction, so \(c-3\notin\Gamma\). Thus $c+e-3\in\apery{\Gamma}{e}^*\cup\{e\}$.
\end{proof}

With this we know the value of the second Feng-Rao distance for the elements in $\Gamma$ up to $c+e-3$. For $c+e-1$ and $c+e-2$ the calculations require more work.

\begin{corollary}\label{coro3}
	For \(\Gamma\) an Arf numerical semigroup with conductor \(c\) and multiplicity \(e\ge 3\) we have
	\[
	\simfrdist{m}=3
	\]
	for all \(m\in \Gamma\) with \(m\le c+e-3\).
	Moreover, 
	\(\simfrdist{c+e-2}=3\) if and only if \(c-2\notin \Gamma\).
\end{corollary}

Let \(\Gamma\) be an Arf numerical semigroup with conductor \(c\) and multiplicity \(e\). Let \(\nxte{\Gamma}{\nxt{e}}=\finset{0}\cup(\nxt{e}+\Gamma)\) with \(\nxt{e}\in\Gamma^*\). The conductor of \(\nxte{\Gamma}{\nxt{e}}\) is then \(\nxt{c}=\nxt{e}+c\). Recall that every Arf numerical semigroup is constructed by applying this procedure several times. The base case is \(\Gamma=\mathbb{N}\), and we have the following. Observe that $\nxte{\mathbb N}{1}=\mathbb N$, and so in the first step we can always omit the case $\nxt{e}=1$.

\begin{lemma}\label{principioordinarios}
	Let \(\nxte{\Gamma}{\nxt{e}}=\finset{0}\cup(\nxt{e}+\mathbb{N})\) with \(\nxt{e}>2\). Then \(\nxt{c}=\nxt{e}\), and
	\[
	\frdist{2}{\nxte{\Gamma}{\nxt{e}}}{\nxt{c}+\nxt{e}-2}=3,
	\]
	and
	\[
	\frdist{2}{\nxte{\Gamma}{\nxt{e}}}{\nxt{c}+\nxt{e}-1}=4.
	\]
\end{lemma}
\begin{proof}
	Clearly, the conductor of $\nxte{\Gamma}{\nxt{e}}$ is $\nxt{e}$. The  equality $\frdist{2}{\nxte{\Gamma}{\nxt{e}}}{\nxt{c}+\nxt{e}-2}=3$ follows from \cref{coro3}, because $\nxt{e}-2\not\in \nxte{\Gamma}{\nxt{e}}$. For the last equality, since
	\(\nxt{c}+\nxt{e}-1\) is the largest irreducible of $\nxte{\Gamma}{\nxt{e}}$, we deduce from \cref{boundsForNDiv} that \(\simfrdist{\nxt{c}+\nxt{e}-1}\geq 4\). But we have
	\[
	\setdiv{\nxte{\Gamma}{\nxt{e}}}{\nxt{c}+\nxt{e}-1,\nxt{c}+\nxt{e}}=
	\finset{0,\nxt{e},\nxt{c}+\nxt{e}-1,\nxt{c}+\nxt{e}},
	\]
	whence \(\simfrdist{\nxt{c}+\nxt{e}-1}=4\).
\end{proof}
Now we go for the general case of the second Feng-Rao distance of $\nxt{c}+\nxt{e}-1$.

\begin{lemma}\label{frdistIncut}
	Let $\Gamma$ be an Arf numerical semigroup with multiplicity $e>1$, and let $\nxt{e}\in \Gamma^*$.Then
	\[
	\frdist{2}{\nxte{\Gamma}{\nxt{e}}}{\nxt{c}+\nxt{e}-1}=
	\begin{cases}
	4 & \text{ if } \nxt{e}=e,\\
	5 & \text{ otherwise}.
	\end{cases}
	\]
\end{lemma}
\begin{proof}
	Take \(m=\nxt{c}+\nxt{e}-1\) and \(m'> \nxt{c}+\nxt{e}-1\) in \cref{boundsForNDiv}, which yields
	\[
	\cardsetdiv{\nxte{\Gamma}{\nxt{e}}}{m'}\leq \cardsetdiv{\nxte{\Gamma}{\nxt{e}}}{\nxt{c}+\nxt{e}-1,m'} \leq \cardsetdiv{\Gamma}{m'} + 1.
	\]
	By taking minimums we obtain
	\[
	\frdist{1}{\nxte{\Gamma}{\nxt{e}}}{\nxt{c}+\nxt{e}} \leq \min\finset{\cardsetdiv{\nxte{\Gamma}{\nxt{e}}}{\nxt{c}+\nxt{e}-1,m'}\mid m'>\nxt{c}+\nxt{e}-1} \leq \frdist{1}{\nxte{\Gamma}{\nxt{e}}}{\nxt{c}+\nxt{e}}+1.
	\]
	Notice that for \(m\ge \nxt{c}+\nxt{e}\) and \(m'>m\), we have
	\[
	\cardsetdiv{\nxte{\Gamma}{\nxt{e}}}{m,m'}\ge \cardsetdiv{\nxte{\Gamma}{\nxt{e}}}{m}+1 \ge \frdist{1}{\nxte{\Gamma}{\nxt{e}}}{\nxt{c}+\nxt{e}}+1.
	\]
	This implies that $\frdist{2}{\nxte{\Gamma}{\nxt{e}}}{\nxt{c}+\nxt{e}-1}$ is attained in \(\cardsetdiv{\nxte{\Gamma}{\nxt{e}}}{\nxt{c}+\nxt{e}-1,m'}\) for some \(m'>m\). Thus
	\[
	\frdist{1}{\nxte{\Gamma}{\nxt{e}}}{\nxt{c}+\nxt{e}}\leq \frdist{2}{\nxte{\Gamma}{\nxt{e}}}{\nxt{c}+\nxt{e}-1} \leq \frdist{1}{\nxte{\Gamma}{\nxt{e}}}{\nxt{c}+\nxt{e}}+1.
	\]
	Since \(e>1\), $\nxte{\Gamma}{\nxt{e}}=\{0,\nxt{e},e+\nxt{e},\ldots\}$. By using \cref{th:CFM}, we obtain \(\frdist{1}{\nxte{\Gamma}{\nxt{e}}}{\nxt{c}+\nxt{e}}=4\), and consequently 
	\[
	\frdist{2}{\nxte{\Gamma}{\nxt{e}}}{\nxt{c}+\nxt{e}-1}\in\{4,5\}.
	\]
	Now \(\frdist{2}{\nxte{\Gamma}{\nxt{e}}}{\nxt{c}+\nxt{e}-1}=4\) if and only if exists $m'\in \mathbb Z$ with \(\nxt{c}+\nxt{e}\leq m'\leq \nxt{c}+\nxt{e}+e-1\) (the upper bound of $m'$ comes from \cref{hastamultiplicidad}) such that \(\cardsetdiv{\nxte{\Gamma}{\nxt{e}}}{\nxt{c}+\nxt{e}-1,m'}=4\). As \(4= \cardsetdiv{\nxte{\Gamma}{\nxt{e}}}{\nxt{c}+\nxt{e}-1,m'}\ge \cardsetdiv{\nxte{\Gamma}{\nxt{e}}}{m'}\ge \frdist{1}{\nxte{\Gamma}{\nxt{e}}}{\nxt{c}+\nxt{e}}=4\), we deduce $\cardsetdiv{\nxte{\Gamma}{\nxt{e}}}{m'}=4$ and \(\nxt{c}+\nxt{e}-1\in \setdiv{\nxte{\Gamma}{\nxt{e}}}{m'}\). Write \(m'=\nxt{c}+\nxt{e}+k\) with \(0\leq k \leq e-1\). This would mean that
	\(m'-(\nxt{c}+\nxt{e}-1)=k+1 \in \nxte{\Gamma}{\nxt{e}}=\finset{0,\nxt{e},e+\nxt{e},\ldots}\). This can only be the case if \(\nxt{e}=e=k-1\). 
\end{proof}

The next step will be describing \(\frdist{2}{\nxte{\Gamma}{\nxt{e}}}{\nxt{c}+\nxt{e}-2}\). 

\begin{lemma}\label{same-e}
	Let $\Gamma$ be an Arf numerical semigroup with multiplicity $e$, and let $\nxt{e}\in \Gamma^*$.
	\begin{itemize}
		\item If \(e=\nxt{e}\), we have
		\[
		\frdist{2}{\nxte{\Gamma}{\nxt{e}}}{\nxt{c}+\nxt{e}-2}=
		\begin{cases}
		3 & \text{ if } \nxt{c}-2\notin \nxte{\Gamma}{\nxt{e}},\\
		4 & \text{ otherwise}.
		\end{cases}
		\]
		
		\item If \(e<\nxt{e}\), then
		\[
		\frdist{2}{\nxte{\Gamma}{\nxt{e}}}{\nxt{c}+\nxt{e}-2}=
		\begin{cases}
		3 & \text{ if } \nxt{c}-2\notin \nxte{\Gamma}{\nxt{e}},\\
		4 & \text{ if } \nxt{c}-2\in \nxte{\Gamma}{\nxt{e}} \text{ and } \nxt{c}=\nxt{e}+2,\\
		5 & \text{ otherwise}.
		\end{cases}
		\]
	\end{itemize}
	
\end{lemma}
\begin{proof}
	By \cref{coro3} we already know that \(\frdist{2}{\nxte{\Gamma}{\nxt{e}}}{\nxt{c}+\nxt{e}-2}=3\) if and only if $\nxt{c}-2\not\in \nxte{\Gamma}{\nxt{e}}$. 
	
	If \(e=\nxt{e}\) and $\nxt{c}-2\in\Gamma$, then \cref{frdistIncut} and \cref{coro3} ensure that $\frdist{2}{\nxte{\Gamma}{\nxt{e}}}{\nxt{c}+\nxt{e}-2}=4$.
	
	Assume that \(e<\nxt{e}\) and \(\nxt{c}-2\in \nxte{\Gamma}{\nxt{e}}\). Then by \cref{frdistIncut}, the only possibilities for \(\frdist{2}{\nxte{\Gamma}{\nxt{e}}}{\nxt{c}+\nxt{e}-2}\) are 4 and 5. By definition, 
	\begin{align*}
		\frdist{2}{\nxte{\Gamma}{\nxt{e}}}{\nxt{c}+\nxt{e}-2} =&
		\min\{\cardsetdiv{\nxte{\Gamma}{\nxt{e}}}{m,m'}\mid \nxt{c}+\nxt{e}-2\le m<m'\} \\ 
		= & \min\{\min\{ \cardsetdiv{\nxte{\Gamma}{\nxt{e}}}{\nxt{c}+\nxt{e}-2,m'}\mid \nxt{c}+\nxt{e}-2<m'\}, \\ 
		& \phantom{\min\{} \min\{  \cardsetdiv{\nxte{\Gamma}{\nxt{e}}}{m,m'}\mid \nxt{c}+\nxt{e}-1\le m<m'\} \}\\ 
		=& \min\{ \min\{ \cardsetdiv{\nxte{\Gamma}{\nxt{e}}}{\nxt{c}+\nxt{e}-2,m'}\mid \nxt{c}+\nxt{e}-2<m'\},\\
		& \phantom{\min\{}\frdist{2}{\nxte{\Gamma}{\nxt{e}}}{\nxt{c}+\nxt{e}-1} \}.
	\end{align*}
	In light of \cref{boundsForNDiv}, $\cardsetdiv{\nxte{\Gamma}{\nxt{e}}}{\nxt{c}+\nxt{e}-2,m'}\ge \cardsetdiv{\nxte{\Gamma}{\nxt{e}}}{\nxt{c}+\nxt{e}-2}+1$, and by \cref{frdistIncut}, $\frdist{2}{\nxte{\Gamma}{\nxt{e}}}{\nxt{c}+\nxt{e}-1}=5$. Consequently
	\[
	\frdist{2}{\nxte{\Gamma}{\nxt{e}}}{\nxt{c}+\nxt{e}-2}\ge \min\finset{5,\cardsetdiv{\nxte{\Gamma}{\nxt{e}}}{\nxt{c}+\nxt{e}-2}+1}.
	\]
	Thus $\frdist{2}{\nxte{\Gamma}{\nxt{e}}}{\nxt{c}+\nxt{e}-2}=4$, if and only if \(\cardsetdiv{\nxte{\Gamma}{\nxt{e}}}{\nxt{c}+\nxt{e}-2}=3\). Now, since \(\nxt{c}-2\in \nxte{\Gamma}{\nxt{e}}\), we have that \(0,\nxt{e},\nxt{c}-2\) and \(\nxt{c}+\nxt{e}-2\) are in $\setdiv{\nxte{\Gamma}{\nxt{e}}}{\nxt{c}+\nxt{e}-2}$. So, necessarily we will have \(\nxt{e}=\nxt{c}-2\). It is easy to see that when \(\nxt{e}=\nxt{c}-2\),  \(\nxte{\Gamma}{\nxt{e}}=\{0,\nxt{e},\nxt{e}+2,\rightarrow\}\) and we have
	\[
	\cardsetdiv{\nxte{\Gamma}{\nxt{e}}}{\nxt{c}+\nxt{e}-2}=3.\qedhere
	\]
\end{proof}

We can summarize the main results of this section as follows. 


\begin{theorem}\label{th:ce}
	Let $\Gamma$ be an Arf numerical semigroup, with multiplicity $e=\rho_{2}$ and conductor $c=\rho_{r}$. 
	\begin{itemize}
		\item If \(e=2\), then for \(m\in \Gamma\) with \(2\leq m \leq c+1\),
		\(
		\simfrdist{m}=\begin{cases} 3 & \text{ if } m=2,\\ 4 & \text{ if } m>2.\end{cases} 
		\)
		\item If \(e>2\), then we have:
		\begin{enumerate}[(1)]
			\item for \(m\in \Gamma\), with \(e\leq m \leq c+e-3\),
			\( \simfrdist{m}=3\);
			\item If $\rho_{3}=2\rho_{2}$, then 
			\[
			\simfrdist{c+e-2}=
			\left\{
			\begin{array}{ll}
			3 & \mbox{if $\rho_{r-1}<c-2$}, \\
			4 & \mbox{if $\rho_{r-1}=c-2$},
			\end{array}
			\right.
			\]
			and $\simfrdist{c+e-1}=4$;
			\item if $\rho_{3}<2\rho_{2}$, then 
			\[
			\delta^{2}_{FR}(c+e-2)=
			\left\{
			\begin{array}{ll}
			3 & \mbox{if $\rho_{r-1}<c-2$}, \\
			4 & \mbox{if $\rho_{r-1}=c-2$ and $r=3$}, \\
			5 & \mbox{if $\rho_{r-1}=c-2$ and $r>3$}, 
			\end{array}
			\right.
			\]
			and 
			\[
			\simfrdist{c+e-1}=
			\begin{cases}
			4 & \text{ if } r=2,\\
			5 & \text{ if } r>2.
			\end{cases}
			\] 
		\end{enumerate}
	\end{itemize}
\end{theorem}

\begin{proof}
	Apply Corollary \ref{coro3} and Lemmas \ref{frdistIncut} and \ref{same-e}, 
	and just notice that $\nxt{\rho}_{3}=e+\nxt{e}$ and that always $\rho_{r-1} \leq c-2$. 
\end{proof}

\subsection{The second Feng-Rao distance for \(m\ge c+e\)}

For a semigroup of the form 
\(\nxte{\Gamma}{\nxt{e}}=\finset{0}\cup\left(\nxt{e}+\Gamma\right)\), we first compute $\delta^{2}_{FR}(m)$ for $m$ in the interval $[\nxt{c},\nxt{c}+\nxt{e}-1]$ with the aid of the previous paragraph, and now we are going to see how to compute it for $m \ge \nxt{c}+\nxt{e}$ in terms of $\Gamma$, so that we can iterate the procedure to get the values in the whole interval $[\nxt{c},2\nxt{c}-1]$. 

\begin{lemma}
	Let $\Gamma$ be an Arf numerical semigroup with multiplicity $e$ and conductor $c$. Let $\nxt{e}\in\Gamma\setminus\{0,e\}$ and $\nxt{c}=c+\nxt{e}$. Then
	\[
	\frdist{2}{\nxte{\Gamma}{\nxt{e}}}{\nxt{c}+\nxt{e}+k}=\frdist{2}{\Gamma}{c+k}+3
	\]
	for all $k\in\mathbb N$.
\end{lemma}
\begin{proof}
	By \cref{generalLema}, we can easily see that
	\[
	\frdist{2}{\Gamma}{c+k}+2\leq \frdist{2}{\nxte{\Gamma}{\nxt{e}}}{\nxt{c}+\nxt{e}+k} \leq \frdist{2}{\Gamma}{c+k}+3.
	\]
	Suppose that the first inequality is actually an equality and let \(k\leq i<j<i+\nxt{e}-1\) be a pair such that $ \frdist{2}{\nxte{\Gamma}{\nxt{e}}}{\nxt{c}+\nxt{e}+k}=\cardsetdiv{\nxte{\Gamma}{\nxt{e}}}{\nxt{c}+\nxt{e}+i,\nxt{c}+\nxt{e}+j}$ (\cref{hastamultiplicidad}). But now
	\begin{align*}
		\frdist{2}{\Gamma}{c+k}+2 &=\frdist{2}{\nxte{\Gamma}{\nxt{e}}}{\nxt{c}+\nxt{e}+k}\\
		&=\cardsetdiv{\nxte{\Gamma}{\nxt{e}}}{\nxt{c}+\nxt{e}+i,\nxt{c}+\nxt{e}+j}\\
		&\ge\cardsetdiv{\Gamma}{c+i,c+j}+2\\
		&\ge \frdist{2}{\Gamma}{c+k}+2.    
	\end{align*}
	This means that the inequalities are all equalities and we will have
	\[
	\cardsetdiv{\nxte{\Gamma}{\nxt{e}}}{\nxt{c}+\nxt{e}+i,\nxt{c}+\nxt{e}+j}=\cardsetdiv{\Gamma}{c+i,c+j}+2,
	\]
	which happens if and only if \(c+\nxt{e}+i\in \setdiv{\Gamma}{c+j}\). Since
	\[
	c+j-(c+\nxt{e}+i)=(j-i)-\nxt{e}\leq 0,
	\]
	this can only be the case if \(j=i+\nxt{e}\). But we will also have
	\[
	\cardsetdiv{\Gamma}{c+i,c+j}=\frdist{2}{\Gamma}{c+k},
	\]
	which is impossible because \(j-i=\nxt{e}>e\) ($\cardsetdiv{\Gamma}{c+i,c+j-e}< \cardsetdiv{\Gamma}{c+i,c+j}$, see the proof of \cref{hastamultiplicidad}).
\end{proof}

We now focus on the case \(\nxt{e}=e\).

\begin{lemma}
	Let $\Gamma$ be a numerical semigroup with multiplicity $e$ and conductor $c$. Let $\nxt{c}=c+e$. Then, for any \(k\in \mathbb N\), the following conditions are equivalent.
	\begin{enumerate}
		\item \(\frdist{2}{\nxte{\Gamma}{{e}}}{\nxt{c}+e+k}=\frdist{2}{\Gamma}{c+k}+2\).
		\item \(\frdist{1}{\Gamma}{c+e+k}=\frdist{2}{\Gamma}{c+k}\).
	\end{enumerate}  
\end{lemma}
\begin{proof}
	Notice that \(\frdist{1}{\Gamma}{c+e+k}\ge \frdist{2}{\Gamma}{c+k}\). To see this, let $p$ be an integer greater than or equal to $k$ such that $\cardsetdiv{\Gamma}{c+e+p}=\frdist{1}{\Gamma}{c+e+k}$. Then 
	\[
	\frdist{2}{\Gamma}{c+k}\leq \cardsetdiv{\Gamma}{c+p,c+e+p}=\cardsetdiv{\Gamma}{c+e+p}=\frdist{1}{\Gamma}{c+e+k}.
	\]
	
	Suppose \(\frdist{2}{\nxte{\Gamma}{{e}}}{\nxt{c}+e+k}=\frdist{2}{\Gamma}{c+k}+2\). Then there must exist integers $i$ and $j$ with
	\(k\leq i <j \leq i+e\) (\cref{hastamultiplicidad}) such that \(\cardsetdiv{\nxte{\Gamma}{e}}{\nxt{c}+e+i,\nxt{c}+e+j}=\frdist{2}{\Gamma}{c+k}+2\). By \cref{generalLema} we have
	\[
	\cardsetdiv{\nxte{\Gamma}{e}}{\nxt{c}+e+i,\nxt{c}+e+j}\geq \cardsetdiv{\Gamma}{c+i,c+j}+2\geq \frdist{2}{\Gamma}{c+k}+2,
	\]
	so these inequalities must all be equalities. This can only happen if \(j=i+e\) (\cref{generalLema}), and then
	\(\cardsetdiv{\nxte{\Gamma}{{e}}}{\nxt{c}+e+i,\nxt{c}+e+j}=\cardsetdiv{\nxte{\Gamma}{{e}}}{\nxt{c}+2e+i}=\frdist{2}{\Gamma}{c+k}+2\), which implies
	\(\cardsetdiv{\Gamma}{c+e+i}=\frdist{2}{\Gamma}{c+k}\)(see \cref{divisorsLem}). So the following inequalities
	\[
	\frdist{2}{\Gamma}{c+k}\leq \frdist{1}{\Gamma}{c+e+k}\leq \cardsetdiv{\Gamma}{c+e+i}
	\]
	must all be equalities.
	
	For the converse, suppose \(\frdist{1}{\Gamma}{c+e+k}=\frdist{2}{\Gamma}{c+k}\). Then there exists an integer $i$ with  \(k\leq i\) such that \(\cardsetdiv{\Gamma}{c+e+i}=\frdist{2}{\Gamma}{c+k}\). But then by \cref{divisorsLem}
	\begin{align*}
		\frdist{2}{\Gamma}{c+k}+2 &=\cardsetdiv{\Gamma}{c+e+i}+2=\cardsetdiv{\nxte{\Gamma}{{e}}}{\nxt{c}+2e+i}\\
		& =\cardsetdiv{\nxte{\Gamma}{{e}}}{\nxt{c}+e+i,\nxt{c}+2e+i}\ge \frdist{2}{\nxte{\Gamma}{{e}}}{\nxt{c}+e+k}.
	\end{align*}
	\cref{divisorsLem} yields \(\frdist{2}{\nxte{\Gamma}{{e}}}{\nxt{c}+e+k}\ge \frdist{2}{\Gamma}{c+k}+2\), and all these inequalities become equalities.
\end{proof}

We summarize now the results of this paragraph in the following theorem.

\begin{theorem}\label{th:further}
	Let $\Gamma$ be an Arf semigroup with conductor \(c\) and multiplicity \(e\). Let $\nxt{e}\in\Gamma^*$ and \(\nxt{c}=e+c\). For every $k\in \mathbb N$,
	\[
	\frdist{2}{\nxte{\Gamma}{\nxt{e}}}{\nxt{c}+\nxt{e}+k}=
	\begin{cases}
	\frdist{2}{\Gamma}{c+k} + 2, \text{ if } \nxt{e}=e \text{ and } \frdist{1}{\Gamma}{c+e+k}=\frdist{2}{\Gamma}{c+k},\\
	\frdist{2}{\Gamma}{c+k} + 3, \text{ otherwise.}
	\end{cases}
	\]
\end{theorem}


As a consequence, as long as we can compute the first and second Feng-Rao distances for $\Gamma$, we can iterate this process to obtain the second Feng-Rao distance for every Arf semigroup by means of a recursive algorithm. 

\begin{algorithm}\label{algor:FR2}
	Feng-Rao distances and numbers of Arf numerical semigroups.
	
	\noindent \textbf{Input:}  The multiplicity sequence $(d_1,\ldots,d_r)$ of an Arf numerical semigroup $\Gamma$.

	\noindent \textbf{Output:} $\simfrnumber{\Gamma}$, $\frdist{1}{\Gamma}{m}$, and $\frdist{2}{\Gamma}{m}$ for all \(m\in \Gamma\).
	
	\begin{enumerate}[1.]
		\item Set \(\Gamma(0)=\mathbb{N}\) with $\frdist{1}{\mathbb{N}}{m}=m+1$ and $\frdist{2}{\mathbb{N}}{m}=m+2$, for $m\in \mathbb{N}$,  $\simfrnumber{\mathbb{N}}=1$.   
		
		\item For $i\in\{1,\ldots,r-1\}$ do
		\begin{itemize}
			\item Set $\Gamma(i):=\nxte{\Gamma(i-1)}{d_{r-i}}$, $e_{i}:=\mathrm e(\Gamma(i))$, and $c_{i}:=\mathrm c(\Gamma(i))$. 
			\item Compute $\simfrnumber{\Gamma(i)}=\min\{d_{r-i},\simfrnumber{\Gamma(i-1)}+1\}$. 
			\item Compute $\frdist{2}{\Gamma({i})}{m}$ for $c \leq m <c+e$ by using Theorem \ref{th:ce}. \item Compute $\frdist{2}{\Gamma({i})}{c+e+k}$ for $k\ge 0$ by using Theorem \ref{th:further}. 
			\item Compute $\frdist{1}{\Gamma({i})}{m}$ for $m \ge 0$ by using Theorem \ref{th:CFM}. 
		\end{itemize}
	\end{enumerate}
\end{algorithm}

Next we illustrate the algorithm with an example.

\begin{example}\label{Ejemplo}
	
	Consider the Arf semigroup $\Gamma=\{0, 12, 24, 32, 36, 40,\to\}$. 
	We will apply Algorithm \ref{algor:FR2} to $\Gamma$, and compute step by step the following (Arf) semigroups. The multiplicity sequence of $\Gamma$ is \((12,12,8,4,4,1)\).
	
	\begin{enumerate}[(1)]
		
		\item We start with \(\Gamma(0)=\mathbb{N}\), with $\frdist{1}{\mathbb{N}}{m}=m+1$ and $\frdist{2}{\mathbb{N}}{m}=m+2$, for $m\in \mathbb{N}$,  $\simfrnumber{\mathbb{N}}=1$.  
		
		\item $\Gamma(1)=\nxte{\mathbb{N}}{4}=\{0,4,\to\}$. It has $c_{1}=e_{1}=4$. Then \(\simfrnumber{\Gamma(1)}=\min\{4,\simfrnumber{\Gamma(0)}+1\}=\min\{4,2\}=2\). The values of \(\frdist{1}{\Gamma(1)}{m}\) and \(\frdist{2}{\Gamma(1)}{m}\) are given in the following table:
		\medskip
		
		\begin{center}
			\begin{tabular}{c|cccccc}
				\hline 
				$m$ & 0&4&5&6&7&$\cdots$ \\
				$\frdist{2}{\Gamma(1)}{m}$ &2&3&3&3&4&$\cdots$ \\
				$\frdist{1}{\Gamma(1)}{m}$ &1&2&2&2&2&$\cdots$ \\
				\hline 
			\end{tabular}
		\end{center}
		\medskip
		
		\item $\Gamma(2) = \nxte{\Gamma(1)}{4}=\{0,4,8,\to\}$. It has $c_2=8, e_2=4$. The second Feng-Rao number is \(\simfrnumber{\Gamma(2)}=\min\{4,\simfrnumber{\Gamma(1)}+1\}=\min\{4,3\}=3\). The values of \(\frdist{1}{\Gamma(2)}{m}\) and \(\frdist{2}{\Gamma(2)}{m}\) are given in two intervals $[8,11]$ and $[12,15]$. 
		
		In the first interval $[c_{2},c_{2}+m_{2}-1]$ we apply Theorem \ref{th:ce}, and we obtain  
		
		\medskip
		
		\begin{center}
			\begin{tabular}{c|cccc}
				\hline 
				$m$ &8&9&10&11 \\
				$\frdist{2}{\Gamma(2)}{m}$ &3&3&3&4 \\
				$\frdist{1}{\Gamma(2)}{m}$ &2&2&2&2 \\
				\hline 
			\end{tabular}
		\end{center}
		\medskip
		
		In the second interval we apply Theorem \ref{th:further}, taking into account that $e_{1}=e_{2}$, and we obtain the results of the following table. Note that in $\Gamma(0)$ we have 
		\[
		\frdist{2}{\Gamma(1)}{4}=3=\frdist{1}{\Gamma(1)}{8}
		\]
		
		\begin{center}
			\begin{tabular}{c|ccccc}
				\hline 
				$m$ &12&13&14&15&$\cdots$ \\
				$\frdist{2}{\Gamma(2)}{m}$ &5&6&6&7&$\cdots$ \\
				$\frdist{1}{\Gamma(2)}{m}$ &4&4&4&4&$\cdots$ \\
				\hline 
			\end{tabular}
		\end{center}
		
		\item $\Gamma({3})= \nxte{\Gamma(2)}{8}=\{0,8,12,16,\to\}$. 
		Now $c_{3}=16$, $e_{3}=8$ and $\simfrnumber{\Gamma(2)}=\min\{8,\simfrnumber{\Gamma(2)}+1\}=\min\{8,4\}=4$. 
		
		For $\Gamma(3)$ we have to consider again two intervals: $[16,23]$ and $[24,31]$. 
		
		In the first interval $[c_{3},c_{3}+m_{3}-1]$ we obtain the following results 
		
		\medskip
		
		\begin{center}
			\begin{tabular}{c|cccccccc}
				\hline 
				$m$ &16&17&18&19&20&21&22&23 \\
				$\frdist{2}{\Gamma(3)}{m}$ &3&3&3&3&3&3&3&5 \\
				$\frdist{1}{\Gamma(3)}{m}$ &2&2&2&2&2&2&2&2 \\
				\hline 
			\end{tabular}
		\end{center}
		
		\medskip
		
		In the second interval we apply again Theorem \ref{th:further}, obtaining the following table. Note that now $e_{3}>e_{2}$. 
		
		\medskip
		
		\begin{center}
			\begin{tabular}{c|ccccccccc}
				\hline 
				$m$ &24&25&26&27&28&29&30&31&$\cdots$ \\
				$\frdist{2}{\Gamma(3)}{m}$ &6&6&6&7&8&9&9&10&$\cdots$ \\
				$\frdist{1}{\Gamma(3)}{m}$ &4&4&4&4&6&6&6&6&$\cdots$\\
				\hline 
			\end{tabular}
		\end{center}
		\medskip
		
		\item $\Gamma(4) =\nxte{\Gamma(3)}{12}=\{0,12,20,24,28,\to\}$. 
		Now $c_{4}=28$, $e_{4}=12$ and $\simfrnumber{\Gamma(4)}=\min\{12,\simfrnumber{\Gamma(3)}+1\}=\min\{12,5\}=5$. 
		
		\item $\Gamma \equiv \Gamma(5) = \nxte{\Gamma(4)}{12}=\{0,12,24,32,36,40,\to\}$. 
		Now $c_{5}=40$, $e_{5}=12$ and $\simfrnumber{\Gamma}=\min\{12,6\}=6$. 
		
	\end{enumerate}
	We proceed in the same way with $\Gamma(4)$ and $\Gamma\equiv\Gamma(5)$, and compute the second Feng-Rao distance in the whole interval $[c_{i},2c_{i}-1]$. 
	In both steps we always sum 3 to the previous semigroup, in $\Gamma(4)$ because $e_{4}>e_{3}$, and in $\Gamma(5)$ because the exception when we sum 2 in Theorem \ref{th:further} never happens. The second Feng-Rao distance for $\Gamma$ is shown in the following table
	
	\medskip
	
	\begin{center}
		\begin{tabular}{c|ccccccccccc}
			\hline 
			$m$ &40&$\cdots$&50&51&52&$\cdots$&62&63&64&$\cdots$&70 \\
			$\frdist{2}{\Gamma}{m}$ &3&$\cdots$&3&4&6&$\cdots$&6&8&9&$\cdots$&9  \\
			\hline 
			\hline 
			$m$ &71&72&73&74&75&76&77&78&79&80&$\cdots$ \\
			$\frdist{2}{\Gamma}{m}$ &11&12&12&12&13&14&15&15&16&17&$\cdots$  \\
			\hline 
		\end{tabular}
	\end{center}
	\medskip
	
	We observe that, in all the steps, the case $m=2c_{i}-1$ matches with the Goppa-like bound $m+1-2g_{i}+\simfrnumber{\Gamma(i)}$. 
	
	We also remark that we are improving the Goppa-like bound given by the second Feng-Rao number, 
	which can be even negative at the beginning of the interval $[c_{i},2c_{i}-1]$. 
	For example, we show below the comparison for $\Gamma_{2}$. 
	
	\medskip
	{\small
		\begin{center}
			\begin{tabular}{c|ccccccccccccccccc}
				\hline 
				$m$ &16&17&18&19&20&21&22&23&24&25&26&27&28&29&30&31&$\cdots$ \\
				$\frdist{2}{\Gamma(2)}{m}$ &3&3&3&3&3&3&3&5&6&6&6&7&8&9&9&10&$\cdots$  \\
				$m+1-2g_2+\simfrnumber{\Gamma(2)}$ &-5&-4&-3&-2&-1&0&1&2&3&4&5&6&7&8&9&10&$\cdots$ \\
				\hline 
			\end{tabular}
		\end{center}
	}
\end{example}

\begin{remark}[Ordinary semigroups]\label{ordinarios}
	
	Let $\Gamma$ be a numerical semigroup such that $c=e$, that is called an \emph{ordinary semigroup}. It is always an Arf numerical semigroup with \(c=\rho_{2}\), that is \(\Gamma=\nxte{\mathbb{N}}{e}\).
	In this case the irreducible elements are precisely 
	\[
	\{e, e+1, \ldots, 2e-1\}.
	\]
	
	Applying \cref{frnumber1} it is easy to see that
	\[
	\simfrnumber{\Gamma}=\cardapery{\Gamma}{1}=2.
	\]
	Thus, if $m\geq 2c-1=2e-1$ we know that 
	\[
	\simfrdist{m}=m+1-2g+\simfrnumber{\Gamma}=m+1-(2e-2)+2= m-(2e-1)+4,
	\]
	because $g=e-1$. 
	
	We can obtain the same result from \cref{th:further}, since \(c=e>1\) we have for \(m\geq c+e =2c\)
	\[
	\frdist{2}{\mathbb{N}_e}{m}=\frdist{2}{\mathbb{N}}{m-2c}+3=m-2c +5=m-(2e-1)+4,
	\]
	and from \cref{principioordinarios}, we get \(\frdist{2}{\mathbb{N}_e}{2c-1}=4= (2c-1)-(2e-1)+4\). Finally, for \(e\leq m<2e-1\), \(\frdist{2}{\mathbb{N}_e}{m}=3.\)
	
	If $\Gamma$ is ordinary with $c=e$, then 
	\[
	\simfrdist{m} =
	\begin{cases}
	3 & \text{ if } c\leq m < 2e-1, \\
	m-(2e-1)+4 & \text{ if } 2e-1 \ge m.
	\end{cases}
	\]
	
	Note that this also applies to the case $e=2$ (the elliptic semigroup). 
	
\end{remark}


\begin{example}
	
	Consider $\Gamma$ the ordinary semigroup with $c=e=6$ and $g=5$. 
	\begin{center}
		\begin{tabular}{c|cccccccc}
			\hline 
			$m$ &6&7&8&9&10&11&12&$\cdots$ \\
			$\simfrdist{m}$ &3&3&3&3&3&4&5&$\cdots$ \\
			\hline 
		\end{tabular}
	\end{center}
\end{example}

\section{Hyperelliptic semigroups}\label{sec:hyperE}

Although section \ref{sec:frdist} gives an algorithm for all Arf semigroups, we study in this section the special case $\Gamma=\langle 2,2g+1\rangle$ with $g$ a positive integer ($g$ is precisely the genus of $\Gamma$), in order to get a closed formula for hyperelliptic semigroups. The conductor is precisely $c=2g$, so that this 
semigroup is symmetric (in fact, these are the only symmetric Arf semigroups). Also \(\Gamma=\nxte{\mathbb{N}}{2,\stackrel{g}{\ldots},2}\) (2 appears $g$ times). Using \cref{Arf-E2} we get 
\[
\simfrnumber{\Gamma} = 2.
\]
Thus, for $m\geq 4g-1$, we have
\[
\simfrdist{m} = m+1-2g+\simfrnumber{\Gamma} = m-2g+3.
\]

Observe that the case of genus equal to one, $\langle 2,3\rangle$, has been considered in the preceding section.

The closed formula for the second Feng-Rao distance is a consequence of the following property.

\begin{lemma}\label{hiperelipticos+2}
	For any hyperelliptic numerical semigroup \(\Gamma=\langle 2,2g+1\rangle\), and any \(k\geq 0\), we have
	\[
	\frdist{1}{\Gamma}{2g +2 +k}=\frdist{2}{\Gamma}{2g+k}.
	\]
\end{lemma}
\begin{proof}
	Recall that \(c=2g\) and \(e=2\).
	We will use induction on \(g\). For \(g=1\) it is easy to see that
	\[
	\frdist{1}{\nxte{\mathbb{N}}{2}}{2+2+k}=\frdist{1}{\mathbb{N}}{k}+2=k+3,
	\]
	while the second Feng-Rao distance is (using \cref{ordinarios})
	\[
	\frdist{2}{\nxte{\mathbb{N}}{2}}{2+k}=2+k-(4-1)+4=k+3.
	\]
	
	Suppose that, the formula holds for \(\Gamma=\langle 2,2(g-1)+1\rangle\), we will prove it for \(\nxte{\Gamma}{2}=\langle 2,2g+1\rangle\). We have, in light of \cref{divisorsLem}, that
	\begin{equation}\label{dfr1}
		\frdist{1}{\nxte{\Gamma}{2}}{2g+2+k}=\frdist{1}{\nxte{\Gamma}{2}}{2(g-1)+2\times 2+k}=\frdist{1}{\Gamma}{2(g-1)+k}+2.
	\end{equation}
	First, suppose that \(k=2+k'\), with $k'\in\mathbb{N}$. By using the induction hypothesis and \cref{th:further}, we deduce
	\[
	\frdist{2}{\nxte{\Gamma}{2}}{2g+k}=\frdist{2}{\nxte{\Gamma}{2}}{2g+2+k'}=
	\frdist{2}{\Gamma}{2(g-1)+k'}+2,
	\]
	but this is equal again by induction hypothesis to $\frdist{1}{\Gamma}{2(g-1)+2+k'}+2$, which is the same as $\frdist{1}{\Gamma}{2(g-1)+k}+2$. Hence $\frdist{1}{\Gamma}{2(g-1)+k}+2=\frdist{2}{\nxte{\Gamma}{2}}{2g+k}$.
	Now by \cref{dfr1}, we obtain
	\[
	\frdist{1}{\nxte{\Gamma}{2}}{2g+2+k}=\frdist{2}{\nxte{\Gamma}{2}}{2g+k}.
	\]
	If \(k\in\{0,1\}\), by \cref{caseMultiplicity2},  \(\frdist{2}{\nxte{\Gamma}{2}}{2g+k}=4\), and by \cref{th:CFM}, $\frdist{1}{\nxte{\Gamma}{2}}{2g+2+k}=6-2=4$ (here $m_2=2g+2-1<2g+2+k\le m_3=2g+4-1$).
\end{proof}

\begin{proposition}
	Let \(\Gamma=\langle 2,2g+1\rangle\) be an hyperelliptic numerical semigroup. Let $k$ be a nonnegative integer smaller than $g$, and \(p\in \{0,1\}\). Then
	\[
	\frdist{2}{\Gamma}{2g+2k+p}=4+2k.
	\]
\end{proposition}
\begin{proof}
	By \cref{th:further} and \cref{hiperelipticos+2} applied $k$ times, we get 
	\[
	\frdist{2}{\Gamma}{2g+2k+p}=\frdist{2}{\nxte{\mathbb{N}}{2,\stackrel{g-k}\ldots,2}}{2(g-k)+p}+2k.
	\]
	Now, \cref{caseMultiplicity2} ensures that $\frdist{2}{\nxte{\mathbb{N}}{2,\stackrel{g-k}\ldots,2}}{2(g-k)+p}=4$.
\end{proof}

For a given numerical semigroup $\Gamma$, we define the \emph{Goppa-like bound} by 
\[
\mathrm{G}_{\Gamma}^{2}(m) := m+1-2g+\simfrnumber{\Gamma}.
\]

We summarize the main results of this section in the following result.

\begin{theorem}
	Let $\Gamma=\langle 2,2g+1\rangle$ with $g\geq 2$,
	\begin{itemize}
		\item For $m=2g+2k+1$, $k\geq 0$, one has $\frdist{2}{\Gamma}{m}=\mathrm G_{\Gamma}^{2}(m)$. 
		\item For $m=2g+2k$, $0 \leq k \leq g-2$, one has $\frdist{2}{\Gamma}{m}=\mathrm G_{\Gamma}^{2}(m)+1$. 
		\item $\frdist{2}{\Gamma}{4g-2}=2g+1=\mathrm G_{\Gamma}^{2}(4g-2)$. 
	\end{itemize}
\end{theorem}

\begin{example}
	
	Consider $\Gamma$ the hyperelliptic semigroup $\langle 2,11\rangle$ with $g=5$ and $c=10$. 
	Computations of the second Feng-Rao distance 
	are summarised in Table \ref{ex:hyperE}.
	
	\begin{center}
		\begin{table*}[!h]
			\begin{tabular}{c|cccccccc}
				\hline 
				$m$ &2&4&6&8&10&11&12&13 \\
				$\simfrdist{m}$ &3&4&4&4&4&4&6&6 \\
				\hline 
				\hline
				$m$ &14&15&16&17&18&19&20&$\cdots$ \\
				$\simfrdist{m}$ &8&8&10&10&11&12&13&$\cdots$ \\
				\hline 
			\end{tabular}
			\caption{$\funsimfrdist$ for the hyperelliptic semigroup $\langle 2,11\rangle$.}\label{ex:hyperE}
		\end{table*}
	\end{center}
	
\end{example}

\section{Computational aspects of the Feng-Rao distance}\label{sec:computation}

Several computer experiments were performed in order to guess the behavior of the second Feng-Rao distance and number for Arf numerical semigroups. We already had some \texttt{GAP} \cite{gap} code for the \texttt{numericalsgps} package \cite{numericalsgps} that was able to compute the Feng-Rao distance of a numerical semigroup. These were used in \cite{intervalos} for the computation of the Feng-Rao numbers of numerical semigroups generated by intervals.

In this section we present an algorithm to find a finite set in which the minimum in the formula
\[
\frdist{r}{\Gamma}{m}=\min\{\cardsimsetdiv{m_1,\ldots,m_r}\mid m\leq m_1<\cdots <m_r, m_i\in \Gamma\},
\]
is attained. 

Let \(\Gamma\) be a numerical semigroup with multiplicity \(e\) and conductor \(c\). Set 
\[\setpossibles{r}{m}=\{(m_1,\ldots,m_r)\in \Gamma^r\mid m\leq m_1<\dots< m_r\}.\] This is the set where we need to find the minimum. Given \(x_1,\ldots,x_r\in \Gamma\) such that \(x_i\geq m\) and \(x_i\neq x_j\) if \(i\neq j\), we denote by \([x_1,\ldots,x_n]\) the unique element in \(\setpossibles{r}{m}\) obtained by sorting the \(x_i\).

We will define recursively a finite subset \(\finpossibles{r}{m}\subset\setpossibles{r}{m}\). First, denote \(u=\max\{m+e-1,c+e-1\}\), and put
\[
\finpossibles{1}{m}=\{m,\ldots,u\}\cap \Gamma.
\]
This is clearly a finite subset of \(\setpossibles{1}{m}\).

Suppose we have defined \(\finpossibles{r-1}{m}\subset \setpossibles{r-1}{m}\), and let \((m_1,\ldots,m_{r-1})\in\finpossibles{r-1}{m}\). Define
\[
X(m;m_1,\ldots,m_{r-1})=\left(\{m_{r-1}+1,\ldots, u\}\cap \Gamma\right)\cup \{m_1+e,\ldots,m_{r-1}+e\}\setminus\{0,\ldots,m_{r-1}\},
\]
which is clearly finite. We also have that for every \(x\in X(m;m_1,\ldots,m_{r-1})\), \(m_{r-1}<x\in \Gamma\) so that \((m_1,\ldots,m_{r-1},x)\in \setpossibles{r}{m}\).
Then set
\begin{align*}
	\finpossibles{r}{m}=\{(m_1,\ldots,m_r)\in \Gamma^r\mid &(m_1,\ldots,m_{r-1})\in \finpossibles{r-1}{m}, \\ &\hspace{3mm}m_r\in X(m;m_1,\ldots,m_{r-1})\}
\end{align*}
which is a subset of \(\setpossibles{r}{m}\), and has finitely many elements.

\begin{proposition}\label{calculo-dfr}
	Let $\Gamma$ be a numerical semigroup, $r>0$ an integer, and $m\in \Gamma$. Then
	\[
	\min\{\cardsimsetdiv{x_1,\ldots,x_r}\mid (x_1,\ldots,x_r)\in \finpossibles{r}{m}\}=
	\min\{\cardsimsetdiv{p_1,\ldots,p_r}\mid (p_1,\ldots,p_r) \in \setpossibles{r}{m}\}.
	\]
\end{proposition}
\begin{proof}
	Clearly, 
	\[\min\{\cardsimsetdiv{x_1,\ldots,x_r}\mid (x_1,\ldots,x_r)\in \finpossibles{r}{m}\}\ge
	\min\{\cardsimsetdiv{p_1,\ldots,p_r}\mid (p_1,\ldots,p_r) \in \setpossibles{r}{m}\}.\]
	
	Suppose that the other inequality does not hold, and let 
	\((p_1,\ldots,p_r)\in \setpossibles{r}{m}\) be such that $\cardsimsetdiv{p_1,\ldots,p_r}< \min\{\cardsimsetdiv{x_1,\ldots,x_r}\mid (x_1,\ldots,x_r)\in \finpossibles{r}{m}\}$. Clearly, $(p_1,\ldots,p_r)\not\in \finpossibles{r}{m}$. We can choose $(p_1,\ldots,p_r)$ to be minimal with respect to the lexicographical ordering fulfilling this condition. 
	
	Since \((p_1,\ldots,p_r)\notin\finpossibles{r}{m}\) we have that either \(p_1\notin \finpossibles{1}{m}\) or there exists \(i\in\{1,\ldots,r-1\}\) such that \((p_1,\ldots,p_i)\in \finpossibles{i}{m}\) but \(p_{i+1}\notin X(m;p_1,\ldots,p_i)\).
	
	If \(p_1\notin\finpossibles{1}{m}=\{m,\ldots,\max\{m+e-1,c+e-1\}\}\cap \Gamma\), that means that \(p_1\geq m+e\) and \(p_1\geq c+e\). So \(p_1-e\geq m\) and \(p_1-e\geq c\). This means \(p_1-e\in\setpossibles{1}{m}\). Now by minimality of \((p_1,\ldots,p_r)\) we should have 
	\[
	\cardsimsetdiv{p_1-e,p_2,\ldots,p_r}\geq \min\{\cardsimsetdiv{x_1,\ldots,x_r}\mid (x_1,\ldots,x_r)\in \finpossibles{r}{m}\}> \cardsimsetdiv{p_1,\ldots,p_r}.
	\]
	which is a contradiction, since \(\simsetdiv{p_1-e}\subset \simsetdiv{p_1}\).
	
	Thus \(p_1\in \finpossibles{1}{m}\). Suppose now that \((p_1,\ldots,p_i)\in \finpossibles{i}{m}\) but \(p_{i+1}\notin X(m;p_1,\ldots,p_i)\). Then, as \(p_{i+1}>p_i\) we must have \(p_{i+1}\geq \max\{m+e-1,c+e-1\}\) and \(p_{i+1}\notin \{m_1+e,\ldots,m_{r-1}+e\}\). This means that \(p_{i+1}-e\geq c\) so \(p_{i+1}-e\in \Gamma\), also \(p_{i+1}-e\geq m\) and \(p_{i+1}-e\neq p_j\) for all \(j\in\{1,\ldots, i\}\). By the minimality of $(p_1,\ldots,p_r)$ we obtain 
	\begin{align*}
		\cardsimsetdiv{[p_1,\ldots,p_i,p_{i+1}-e,p_{i+2},\ldots,p_r]}&\geq 
		\min\{\cardsimsetdiv{x_1,\ldots,x_r}\mid (x_1,\ldots,x_r)\in \finpossibles{r}{m}\}\\ &>\cardsimsetdiv{p_1,\ldots,p_r},
	\end{align*}
	which is again a contradiction. 
\end{proof}

Observe that since $\finpossibles{r}{m}$ can be constructed recursively and has finitely many elements, \cref{calculo-dfr} provides a computational procedure to calculate $\frdist{r}{\Gamma}{m}$.


\section{Examples and conclusions}\label{sec:conclusiones}

As we told in the Example~\ref{Ejemplo}, the exact value of the second Feng-Rao distance is a much better estimate for the second Hamming weight than the Goppa-like given by the second Feng-Rao number. In this sense, the results of this paper strongly improve those of the paper~\cite{inductivos} for AG codes coming from inductive semigroups, like those constructed from the tower of function fields given in~\cite{GS}. Notice that Arf semigroups are not symmetric (except for the hyperelliptic case), so that the equality between generalized Feng-Rao distances and Goppa-like bounds is very rare. 

Let us recall now the definition of the generalized Hamming weights. First, the support of a linear code $C$ is defined as
\[
{\rm supp}(C):=\{i \mid c_{i}\neq 0\;\;\mbox{for some ${\bf c}\in C$}\}.
\]
Thus, the $r$th generalized Hamming weight of $C$ is given by
\[
{\mathrm d}_{r}(C):=\min\{\sharp\,{\rm supp}(C')\mid \mbox{$C'\preceq C$ with ${\rm dim}(C')=r$}\}, 
\]
where $C'\preceq C$ denotes a linear subcode $C'$ of $C$. In this paper we focus on $r=2$. 
Thus, we know that 
\[
\mathrm d_{2}(C_{m}) \geq \simfrdist{m+1} \geq m+2-2g+\simfrnumber{\Gamma}
\]
for a one-point AG code constructed from an algebraic curve of genus $g$ whose involved  Weierstrass semigroup is $\Gamma$, as long as $m$ is larger than or equal to the conductor of $\Gamma$ (see the details in~\cite{HvLP}). This is called the Goppa-like bound, and we denote it by ${\rm GLB(m)}$. 

The results in~\cite{inductivos} improve previous bounds of Pellikaan in~\cite{K-P} or the Griesmer order bound (see \cite{HKM} and \cite{D}), so that the results of this paper also improve them, as a consequence. More precisely, Pellikaan bound in~\cite[Theorem 2.8]{K-P}  for $r=2$ states that 
\[
\mathrm d_{2}(C_{m}) \geq \delta_{\rm FR}(m+2).
\]
On the other hand, the Griesmer order bound for $r=2$ yields
\[
\mathrm d_{2}(C_{m}) \geq {\rm GOB}(m+1) := \delta_{\rm FR}(m+1) + 
\left\lceil\displaystyle\frac{\delta_{\rm FR}(m+1)}{q}\right\rceil,
\]
where $q$ is the size of the finite field underlying the code $C_{m}$. 

We apply now our results to AG codes coming from the tower of function fields given in~\cite{GS}. Let us recall the definitions, and leave the details also to~\cite{inductivos}. 

Consider the tower of function fields $({\mathcal T}_{n})$ over $\mathbb{F}_{q^2}$, 
where ${\mathcal T}_{1}=\mathbb{F}_{q^{2}}(x_{1})$ and for $n\geq 2$, ${\mathcal T}_{n}$ is obtained from ${\mathcal T}_{n-1}$ by adjoining a new element $x_{n}$ satisfying the equation
\[
x_{n}^{q}+x_{n}=\frac{x_{n-1}^{q}}{x_{n-1}^{q-1}+1}.
\]
This tower attains the Drinfeld-Vl\u{a}du\c{t} bound (see~\cite{HvLP}). As a consequence, error-correcting AG codes construncted from this tower are very interesting because of their excellent asymptotical behaviour. 

Let $Q_{n}$ be the rational place on ${\mathcal T}_{n}$ that is the unique pole of $x_{1}\,$. 
It is known that the Weierstrass semigroups $\Gamma^{n}$ of ${\mathcal T}_{n}$ at $Q_{n}$ are as follows: $\Gamma^{1}=\mathbb{N}$, and for $n\geq 2$,
\[
\Gamma^{n} = q \cdot \Gamma^{n-1} \cup \{ m\in\mathbb{N} \mid m\geq c_n\},
\]
where 
\[
c_{n}=\left\{\begin{array}{ll}
q^{n}-q^{\frac{n+1}{2}} & \mbox{if $n$ is odd}, \\
q^{n}-q^{\frac{n}{2}} & \mbox{if $n$ is even}. \end{array}\right. 
\]
Thus, these numerical semigroups $\Gamma^{n}$ are inductive, and they are Arf in particular (see~\cite{Arf-IEEE}). In fact, you can see in \cite{inductivos} a description of $\Gamma^{n}$ with $n\geq 2$ as a disjoint union of sets $\Lambda^i$ as follows. Write $n=2k+b$ with $k\geq 1$ and $b=0,1$, and set: 

\begin{itemize}
	\item $\Lambda^0=\{0,q^{n-1},2q^{n-1}, \ldots, (q-1) \cdot q^{n-1}\}$,
	\item $\Lambda^1=(q-1) q^{n-1}+\{q^{n-3},2q^{n-3},\ldots, (q-1)q \cdot q^{n-3}\}$,
	\item $\Lambda^2=[(q-1) q^{n-1} + (q-1) q^{n-2}]+
	\{q^{n-5},2q^{n-5},\ldots, (q-1)q^{2} \cdot q^{n-5}\}$,
	\item \dots 
	\item $\Lambda^i=(q-1)[q^{n-1}+\cdots+q^{n-i}]+
	\{q^{n-1-2i},2q^{n-1-2i},\ldots, (q-1)q^{i} \cdot q^{n-1-2i}\}$,
	\item \dots 
	\item $\Lambda^{k-1}=(q-1)[q^{n-1}+\cdots+q^{n-k+1}]+
	\{q^{b+1},2q^{b+1},\ldots, (q-1)q^{k-1} \cdot q^{b+1}\}$,
	\item $\Lambda^{k}=(q-1)[q^{n-1}+\cdots+q^{n-k}]+\mathbb N^{\ast}$. 
\end{itemize}






Thus, the semigroup $\Gamma^{n}$ can be easily recovered from the Algorithm \ref{algor:FR2}. 

We show now several examples comparing the Pellikaan bound, the Griesmer order bound, the Goppa-like bound with the second Feng-Rao number, and the bound from the second Feng-Rao-bound. Note that the AG codes only make sense for $m>2g-2$, and not only $m\geq c$ (see~\cite{HvLP}). 

In both examples, we consider the dual one-point AG code $C_{m}$ over $\mathbb{F}_{q^{2}}$ defined by the divisor $G=mQ_{n}$, $Q_{n}$ the rational place defined above (see~\cite{HvLP} for further details). 

\begin{example}
	
	Consider the 5$th$ floor of the above tower of function fields for $q=3$ (note that the codes are constructed over the finite field $\mathbb{F}_{9}$\/). Thus $n=5$, $k=2$ and $b=1$, and the semigroup $\Gamma^{5}$, with conductor $c=216$ and genus $g=208$ is decomposed into the following sets: 
	
	\begin{itemize}
		
		\item $\Lambda^{0} = \{ 0, 81, 162\}$, 
		
		\item $\Lambda^{1} = 162 + \{ 9, 18, 27, 36, 45, 54\}$, 
		
		\item $\Lambda^{2} = 216 + \mathbb N^{\ast}$.
		
	\end{itemize}
	
	Thus, we have to perform Algorithm \ref{algor:FR2} with successive translations 
	9, 9, 9, 9, 9, 9, 81, 81. The results for $m\geq 415$ are given in Table~\ref{q3n5}. Notice that, for the Goppa-like bound, the second Feng-Rao number is $E_{2}=9$, and that for the Griesmer order bound the size of the finite field is 9. 
	
	\medskip 
	
	\begin{small}
		\begin{center}
			\begin{table}[!t]
				\begin{tabular}{|l|ccccccccccc|}
					\hline 
					$m$ & $[415,420]$ & 421 & 422 & 423 & 424 & 425 & 426 & 427 & 428 & 429 & 430 \\
					\hline 
					$\simfrdist{m+1}$ & 18 & 18 & 19 & 20 & 21 & 22 & 23 & 24 & 24 & 24 & 25 \\
					\hline 
					$\delta_{\rm FR}(m+1)$ & 14 & 14 & 16 & 16 & 16 & 16 & 16 & 16 & 16 & 16 & 16 \\
					\hline 
					${\rm GOB}(m+1)$ & 16 & 16 & 18 & 18 & 18 & 18 & 18 & 18 & 18 & 18 & 18 \\
					\hline 
					$\delta_{\rm FR}(m+2)$ & 14 & 16 & 16 & 16 & 16 & 16 & 16 & 16 & 16 & 16 & 17 \\
					\hline 
					${\rm GLB}(m)$ & $\leq 15$ & 16 & 17 & 18 & 19 & 20 & 21 & 22 & 23 & 24 & 25 \\
					\hline 
				\end{tabular}
				\caption{Parameters of the code $C_{m}$ from inductive tower, for $q=3$ and $n=5$.}\label{q3n5}
			\end{table}
		\end{center}
	\end{small}
	
\end{example}

\begin{example}
	
	Consider the 8$th$ floor of the above tower of function fields for $q=2$, the codes being constructed over $\mathbb{F}_{4}$. Thus $n=8$, $k=4$ and $b=0$, and the semigroup $\Gamma^{8}$, with conductor $c=240$ and genus $g=225$ is decomposed into the following sets: 
	
	\begin{itemize}
		
		\item $\Lambda^{0} = \{ 0, 128\}$, 
		
		\item $\Lambda^{1} = 128 + \{ 32, 64\}$, 
		
		\item $\Lambda^{2} = 192 + \{ 8, 16, 24, 32\}$, 
		
		\item $\Lambda^{3} = 224 + \{ 2, 4, 6, 8, 10, 12, 14, 16\}$, 
		
		\item $\Lambda^{4} = 240 + \mathbb N^{\ast}$.
		
	\end{itemize}
	
	Thus, we have to perform the Algorithm \ref{algor:FR2} with successive translations 
	2, 2, 2, 2, 2, 2, 2, 2, 8, 8, 8, 8, 32, 32, 128. The results for $m \geq 449$ are given in Table~\ref{q2n8}. Note that now the size of the finite field is 4, and $E_{2}=9$. 
	
	\medskip 
	
	\begin{small}
		\begin{center}
			\begin{table}[!t]
				\begin{tabular}{|l|ccccccccccc|}
					\hline 
					$m$ & $[449,453]$ & 454 & $[455,456]$ & $[457,460]$ & 461 & 462 & 463 & 464 & 465 & 466 & 467 \\
					\hline 
					$\simfrdist{m+1}$ & 17 & 17 & 19 & 21 & 23 & 23 & 25 & 25 & 27 & 27 & 29 \\
					\hline 
					$\delta_{\rm FR}(m+1)$ & 12 & 12 & 14 & 14 & 14 & 14 & 16 & 16 & 18 & 18 & 20 \\
					\hline 
					${\rm GOB}(m+1)$ & 15 & 15 & 18 & 18 & 18 & 18 & 20 & 20 & 23 & 23 & 25 \\
					\hline 
					$\delta_{\rm FR}(m+2)$ & 12 & 14 & 14 & 14 & 14 & 16 & 16 & 18 & 18 & 20 & 20 \\
					\hline 
					${\rm GLB}(m)$ & $\leq 14$ & 15 & $\leq 17$ & $\leq 21$ & 22 & 23 & 24 & 25 & 26 & 27 & 28 \\
					\hline 
					\hline 
					$m$ & 468 & 469 & 470 & 471 & 472 & 473 & 474 & 475 & 476 & 477 & 478 \\
					\hline 
					$\simfrdist{m+1}$ & 29 & 31 & 31 & 33 & 33 & 35 & 35 & 37 & 37 & 38 & 39 \\
					\hline 
					$\delta_{\rm FR}(m+1)$ & 20 & 22 & 22 & 24 & 24 & 26 & 26 & 28 & 28 & 30 & 30 \\
					\hline 
					${\rm GOB}(m+1)$ & 25 & 28 & 28 & 30 & 30 & 33 & 33 & 35 & 35 & 38 & 38 \\
					\hline 
					$\delta_{\rm FR}(m+2)$ & 22 & 22 & 24 & 24 & 26 & 26 & 28 & 28 & 30 & 30 & 31 \\
					\hline 
					${\rm GLB}(m)$ & 29 & 30 & 31 & 32 & 33 & 34 & 35 & 36 & 37 & 38 & 39 \\
					\hline 
				\end{tabular}
				\caption{Parameters of the code $C_{m}$ from inductive tower, for $q=2$ and $n=8$.}\label{q2n8}
			\end{table}
		\end{center}
	\end{small}

\end{example}

As a conclusion, in sight of the above examples it is clear that the results of this paper are a kind of generalization of those in~\cite{Arf-IEEE} for the second Feng-Rao distance, in the sense that this distance is constant in large bursts, corresponding to the intervals $[c+\rho_{i},c+\rho_{i+1}-1]$ or subintervals of them.

\end{document}